\documentclass[10pt, draftcls, onecolumn]{IEEEtran}

\usepackage{amsmath,times,amssymb,amsthm,epsfig,nicefrac,color}
\usepackage{amsfonts,graphics,epsfig,graphics,caption,subcaption}
\usepackage{verbatim,makeidx}
\usepackage[mathcal]{euscript} 
\usepackage{url}
\usepackage{dsfont}
\usepackage{cite,calc}
\usepackage{euler}
\usepackage{algorithm}
\usepackage{algorithmic}
\usepackage{pbox}
\usepackage{bbm}
\usepackage{algorithm}
\usepackage{algorithmic}
\usepackage{colortbl}
\usepackage{float}
\usepackage{graphicx}
\usepackage{pbox}

\usepackage{todonotes}

\newcounter{ALC@tempcntr}% Temporary counter for storage

\theoremstyle{plain}% default
\newtheorem{theorem}{Theorem}[section]
\newtheorem{proposition}{Proposition}[section]
\newtheorem{lemma}{Lemma}[section]
\newtheorem{corollary}{Corollary}[section]
\newtheorem{claim}{Claim}[section]
\newtheorem{definition}{Definition}

\theoremstyle{definition}
\newtheorem{example}{Example}%[section]
\newtheorem{construction}{Construction}%[section]
\newtheorem{remark}{Remark}%[section]

\theoremstyle{remark}
\newcommand{\beq}{\begin{eqnarray}}
\newcommand{\eeq}{\end{eqnarray}}

\newcommand{\field}[1]{\mathbb{#1}}

\newcommand{\F}{\field{F}}
\newcommand{\B}{\field{B}}
\newcommand{\Lf}{\field{L}}

\newcommand{\cC}{{\cal C}}

\newfont{\bbb}{msbm10 scaled 500}

\newfont{\bb}{msbm10 scaled 1100}

\newcommand{\G}{\mbox{\bb G}}

\newcommand{\cv}{{\bf c}}

\newcommand{\ev}{{\bf e}}

\newcommand{\xv}{{\bf x}}

\newcommand{\Am}{{\bf A}}

\newcommand{\Em}{{\bf E}}

\newcommand{\Hm}{{\bf H}}
\newcommand{\Id}{{\bf I}}

\newcommand{\Pm}{{\bf P}}

\newcommand{\Um}{{\bf U}}

\newcommand{\Ac}{{\cal A}}

\newcommand{\Cc}{{\cal C}}
\newcommand{\Dc}{{\cal D}}

\newcommand{\Hc}{{\cal H}}

\newcommand{\Kc}{{\cal K}}

\newcommand{\Nc}{{\cal N}}

\newcommand{\Pc}{{\cal P}}

\newcommand{\Rc}{{\cal R}}
\newcommand{\Sc}{{\cal S}}
\newcommand{\Tc}{{\cal T}}

\newcommand{\Vc}{{\cal V}}

\DeclareMathOperator{\rank}{rank}
\DeclareMathOperator{\rk}{rank}

\DeclareMathOperator{\Diag}{Diag}

\newcommand{\remove}[1]{}

\newcommand{\E}{\mathrm{E}}

\newcommand{\dist}{d_\mathrm{H}}

\theoremstyle{definition}

\theoremstyle{remark}

\newcommand{\latexe}{{\LaTeX\kern.125em2%
                      \lower.5ex\hbox{$\varepsilon$}}}

\chardef\bslash=`\\	

\makeatletter		
\def\square{\RIfM@\bgroup\else$\bgroup\aftergroup$\fi
\vcenter{\hrule\hbox{\vrule\@height.6em\kern.6em\vrule}
\hrule}\egroup}\makeatother\makeindex

\definecolor{OXO-emph}{RGB}{153,0,0}
\newcommand\ceilb[1]{\left\lceil #1 \right\rceil}
\newcommand\ceilbb[1]{\bigl\lceil #1 \bigr\rceil}
\newcommand\floorb[1]{\left\lfloor #1 \right\rfloor}
\newcommand\floorbb[1]{\bigl\lfloor #1 \bigr\rfloor}
\DeclareMathAlphabet{\mathpzc}{OT1}{pzc}{m}{it}
\newcolumntype{?}{!{\vrule width 1.5pt}}

\usepackage{enumitem}
\setlist[itemize]{leftmargin=0.2in}
\setlist[enumerate]{leftmargin=0.2in}

 \IEEEoverridecommandlockouts

\begin{document}
\sloppy

%%%%%%%%%%%%%%%%%%%%%%%%%%%%%%%%%%%%%%%%%%%%%%%%%%%%%%%%%%%%%%%
%%%%%%%%%%%%%%%%%%%%%%%%%%%%%%%%%%%%%%%%%%%%%%%%%%%
%% Title and Authors
%%%%%%%%%%%%%%%%%%%%%%%%%%%%%%%%%%%%%%%%%%%%%%%%%%%
%%%%%%%%%%%%%%%%%%%%%%%%%%%%%%%%%%%%%%%%%%%%%%%%%%%%%%%%%%%%%%%

%\title{$\epsilon$-MSR Codes with Small Sub-packetization}
\title{MDS Code Constructions with Small Sub-packetization and \\ Near-optimal Repair Bandwidth}
\author{
\IEEEauthorblockN{Ankit Singh Rawat, Itzhak Tamo~\IEEEmembership{Member,~IEEE,} Venkatesan Guruswami, and Klim Efremenko}%
\thanks{This research was supported in part by NSF grants CCF-0963975, CCF-1422045, and CCF-1563742. This paper was presented in parts at the 2016 Information Theory and Applications Workshop~\cite{TK16}, 2017 ACM-SIAM Symposium on Discrete Algorithms~\cite{GR17} and 2017 IEEE International Symposium on Information Theory~\cite{RTGE17}.}
\thanks{A.~S.~Rawat is with the Research Laboratory of Electronics, MIT, Cambridge, MA 02139, USA (e-mail: asrawat@mit.edu). Part of this work was done when the author was with the Computer Science Department, Carnegie Mellon University, Pittsburgh, PA 15213, USA.}
\thanks{I.~Tamo is with the Department of Electrical Engineering - Systems, Tel Aviv University, Ramat Aviv 69978, Israel (e-mail: tamo@post.tau.ac.il).}
\thanks{V.~Guruswami is with the Computer Science Department, Carnegie Mellon University, Pittsburgh, PA 15213, USA (e-mail: venkatg@cs.cmu.edu).}
\thanks{K.~Efremenko is with the Department of Computer Science, Tel Aviv University, Ramat Aviv 69978, Israel (e-mail: klimefrem@gmail.com).}
}

\maketitle

\allowdisplaybreaks
\maketitle

\begin{abstract}
A code $\mathcal{C} \subseteq \mathbb{F}^n$ is a collection of $M$ codewords where $n$ elements (from the finite field $\mathbb{F}$) in each of the codewords are referred to as code blocks. Assuming that $\mathbb{F}$ is a degree $\ell$ extension of a smaller field $\B$, the code blocks are treated as $\ell$-length vectors over the base field $\mathbb{B}$. Equivalently, the code is said to have the sub-packetization level $\ell$. This paper addresses the problem of constructing MDS codes that enable exact reconstruction (repair) of each code block by downloading small amount of information from the remaining code blocks. The total amount of information flow from the remaining code blocks during this reconstruction process is referred to as repair-bandwidth of the underlying code. The problem of enabling exact reconstruction of a code block with small repair bandwidth naturally arises in the context of distributed storage systems as the node repair problem~\cite{dimakis}. The constructions of exact-repairable MDS codes with optimal repair-bandwidth require working with large sub-packetization levels, which restricts their employment in practice.

This paper presents constructions for MDS codes that simultaneously provide both small repair bandwidth and small sub-packetization level. In particular, this paper presents two general approaches to construct exact-repairable MDS codes that aim at significantly reducing the required sub-packetization level at the cost of slightly sub-optimal repair bandwidth. The first approach provides MDS codes that have repair bandwidth at most twice the optimal repair-bandwidth. Additionally, these codes also have the smallest possible sub-packetization level $\ell = O(r)$, where $r$ denotes the number of parity blocks. This approach is then generalized to design codes that have their repair bandwidth approaching the optimal repair-bandwidth at the cost of graceful increment in the required sub-packetization level. The second approach provides ways to transform an MDS code with optimal repair-bandwidth and large sub-packetization level into a longer MDS code with small sub-packetization level and near-optimal repair bandwidth. For a given number of parity blocks, the codes constructed using this approach have their sub-packetization level scaling logarithmically with the code length. In addition, the obtained codes require field size only linear in the code length and ensure load balancing among the intact code blocks in terms of the information downloaded from these blocks during the exact reconstruction of a code block.
\end{abstract}

%\begin{keywords}
%Codes for distributed storage, regenerating codes, sub-packetization level, repair bandwidth, cut-set bound.
%\end{keywords}

%%%%%%%%%%%%%%%%%%%%%%%%%%%%%%%%%%%%%%%%%%%%%%%%%%%%%%%%%%%%%%%
%%%%%%%%%%%%%%%%%%%%%%%%%%%%%%%%%%%%%%%%%%%%%%%%%%%
%% Introduction
%%%%%%%%%%%%%%%%%%%%%%%%%%%%%%%%%%%%%%%%%%%%%%%%%%%
%%%%%%%%%%%%%%%%%%%%%%%%%%%%%%%%%%%%%%%%%%%%%%%%%%%%%%%%%%%%%%%

\section{Introduction.}
\label{sec:intro}

Maximum distance separable (MDS) codes are considered to be an attractive solution for information storage as they operate at the optimal storage vs. reliability trade-off given by the Singleton bound~\cite{MacSlo}. For a given amount of information to be stored and available storage space, the MDS codes can tolerate the maximum number of worst case failures without losing the stored information. However, the applicability of the MDS codes in modern storage systems also depends on their ability to efficiently reconstruct parts of a codeword from the rest of the codeword. Consider a distributed storage system which employs an MDS code to store information over a network of storage nodes such that each storage node stores a small part of a codeword from the MDS code. Exact reconstruction (repair) of the content stored in a node with the help of the content stored in the remaining nodes is useful to reinstate the system in the event of a permanent node failure. Similarly, this also enables access to the information stored on a temporarily unavailable node with the help of available nodes in the system. Therefore, among all MDS codes, the ones with more efficient exact repair mechanisms are preferred for deployment in modern distributed storage systems.

In \cite{dimakis}, Dimakis et al. study the repair problem in distributed storage systems and introduce {\em repair bandwidth}, the amount of data downloaded during a node repair, as a metric to measure the efficiency of a node repair mechanism. In particular, they consider a setup where an MDS code is employed to encode a file containing $k$ symbols over a finite field $\F$ to a codeword comprising $n$ symbols over $\F$. These $n$ symbols are then stored on $n$ storage nodes. For an MDS code, it is straightforward to achieve a repair bandwidth of $k$ symbols (over $\F$) by contacting any $k$ remaining nodes and downloading the $k$ distinct symbols stored on these nodes. This follows from the fact that any $k$ symbols of a codeword from an MDS codes are sufficient to reconstruct the {\em entire} codeword. Note that the repair bandwidth of $k$ symbols (over $\F$) is the best possible if we are allowed to contact only $k$ remaining storage nodes during the repair process. Furthermore, it is not possible to reconstruct a code symbol by contacting less than $k$ remaining code symbols of a codeword in an MDS code. This motivates Dimakis et al. to look for potentially lowering the repair bandwidth for repair of a single node by contacting $t \geq k$ remaining nodes in the system and downloading partial data stored on each of the contacted nodes. 

For the repair of a failed node, each code symbol of the codeword is viewed as an $\ell$ length vector (code block) over a subfield $\B$, where $\ell$ is referred to as {\em sub-packetization level} or {\em node size}. Given this vector representation of the MDS code, the repair bandwidth of an MDS code is lower bounded by~\cite{dimakis,HLKB15} 
\begin{align}
\label{eq:cut_set_d}
\Big(\frac{t}{t - k + 1}\Big)\cdot \ell~~~\text{symbols (over $\B$)}.
\end{align}
In the particular case, when $t = n -1$, i.e., all the remaining nodes in the system are contacted during the repair process, the bound on the repair bandwidth reduces to
$
\Big(\frac{n - 1}{n - k}\Big)\cdot \ell~\text{symbols (over $\B$)}.
$

The bound in \eqref{eq:cut_set_d} is referred to as the cut-set bound in the literature. The problem of designing {exact-repairable} MDS codes with the optimal repair bandwidth (cf.~\eqref{eq:cut_set_d}) has led to many novel code designs that are proposed in \cite{RSK11, CJMRS13, zigzag13, PapDimCad13, SAK15, RKV16a, GoparajuFV16} and references therein. In \cite{YeB16a}, Ye and Barg present the first fully explicit construction of MDS codes with optimal repair bandwidth for all values of the system parameters $n$, $k$, and $t$. Further explicit constructions of MDS codes with optimal repair bandwidth are presented in \cite{YeB16b, SVK16}. Note that as the number of the nodes contacted during the repair process $t$ gets larger, the optimal repair bandwidth defined by the cut-set bound becomes significantly smaller than the naive repair bandwidth of $k$ symbols (over $\F$) or $k\ell$ symbols (over $\B$). 

Here, we note that the aforementioned existing constructions for high rate MDS codes with the optimal repair bandwidth require a large sub-packetization level $\ell$. In this paper, we aim to construct MDS codes that have both small repair bandwidth and small sub-packetization level. Towards this, we relax the requirement that the underlying MDS code attains the cut-set bound. In particular, we design MDS codes with repair bandwidth arbitrarily close to the cut-set bound. This small loss in terms of repair bandwidth optimality results in a significantly large benefit in terms of the required sub-packetization level. Thus, this paper essentially explores a trade-off between the sub-packetization level $\ell$ and the repair bandwidth for MDS codes. 

\begin{figure*}
        \centering
        \begin{subfigure}[b]{0.48\textwidth}
                \centering
                \includegraphics[width=0.8\textwidth]{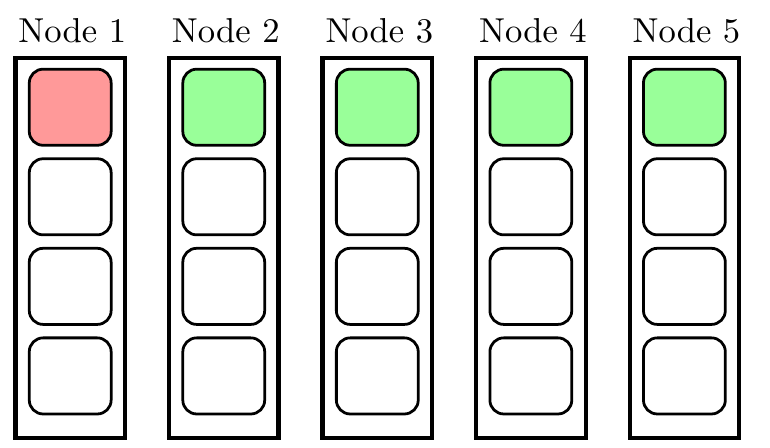}
                \caption{}
                \label{fig:good_degraded}
        \end{subfigure}%
       ~%add desired spacing between images, e. g. ~, \quad, \qquad etc.
          %(or a blank line to force the subfigure onto a new line)
        \begin{subfigure}[b]{0.48\textwidth}
                \centering
                \includegraphics[width=0.8\textwidth]{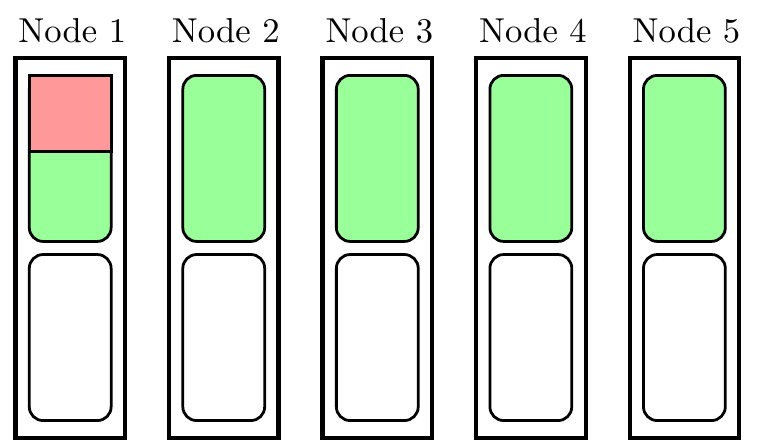}
                \caption{}
                \label{fig:bad_degraded}
        \end{subfigure}
        \caption{Illustration of efficient degraded reads in systems employing codes with small sub-packetization level. In Figure~\ref{fig:good_degraded}, we have a system with five storage nodes storing four different systematic codewords associated with independent data. For example, five colored symbols correspond to one of these four codewords. In the event when the content on the first node is unavailable (e.g., due to a parallel read being served by the node $1$), the required (red colored) information block can be recovered by treating the block as a failed block and regenerating it using the remaining (green colored) blocks of the associated codeword in a bandwidth efficient manner. Now consider the setting in Figure~\ref{fig:bad_degraded}, where a code with large sub-packetization is employed. Here, the system stores two distinct code words. We again assume that the first node is unavailable to serve a read request for the information highlighted by red color. Since the required information only constitute part of a code block, recovering the desired information by invoking the repair of its associated block is not very efficient as it would incur larger bandwidth as compared to the setting in Figure~\ref{fig:good_degraded}.}
        \label{fig:degraded}
\end{figure*}

MDS codes that provide small sub-packetization level in addition to having small repair bandwidth are of great practical importance in distributed storage systems. The smaller sub-packetization level leads to easier system implementation. Another obvious advantage of working with small sub-packetization level is that it provides a system designer with greater flexibility in terms of selecting various system parameters. As an example, consider a scenario where an MDS code requires a large sub-packetization level, e.g., say $\ell \geq 2^n$. This implies that using storage nodes (disks) with storage capacity of $\ell$ symbols (over $\B$), one can only design a storage system with at most $\log_2 \ell$ nodes. Therefore, larger sub-packetization level can lead to a reduced design space in terms of various system parameters. Here, we list some of the additional advantages of having small sub-packetization level that might be specific to the way various large scale distributed storage systems operate.

\begin{itemize}
\item A code with larger sub-packetization level makes management of meta-data difficult. Here, we use meta-data to refer to the descriptions of the code (e.g., its parity-check matrix in case of a linear code) and repair mechanisms associated with different code blocks. Employing a code that requires large sub-packetization level implies that such descriptions have large size which makes it non-trivial and (or) resource inefficient to make this information available in a distributed storage setup.
\item Interestingly, a code with small sub-packetization level enables bandwidth efficient (on-the-fly) accesses to missing small files (information blocks) stored on a distributed storage system by performing degraded reads~\cite{Khan12}. Assuming that the sub-packetization level of the underlying MDS codes is smaller than the size of the files accessed by users, one can obtain a missing file (information block) by utilizing the bandwidth efficient repair mechanism of the MDS code. Towards this, one can treat the missing files as failed (systematic) code blocks in the associated codewords and access the file by performing bandwidth efficient repair of the failed (systematic) code blocks comprising the missing file of interest (cf.~Figure~\ref{fig:degraded}). Here, we note that a file may be considered missing due to transient unavailability of some storage nodes in the system.
\item As illustrated in Figure~\ref{fig:repair_balance}, in a large scale distributed storage system, small sub-packetization might allow us to distribute the codewords corresponding to independently coded files among multiple nodes. In some settings, this may be advantageous as it would distributed the load of providing information for the repair of a failed node among a large number of nodes.
\end{itemize}

\begin{figure*}
        \centering
         \includegraphics[width=0.8\textwidth]{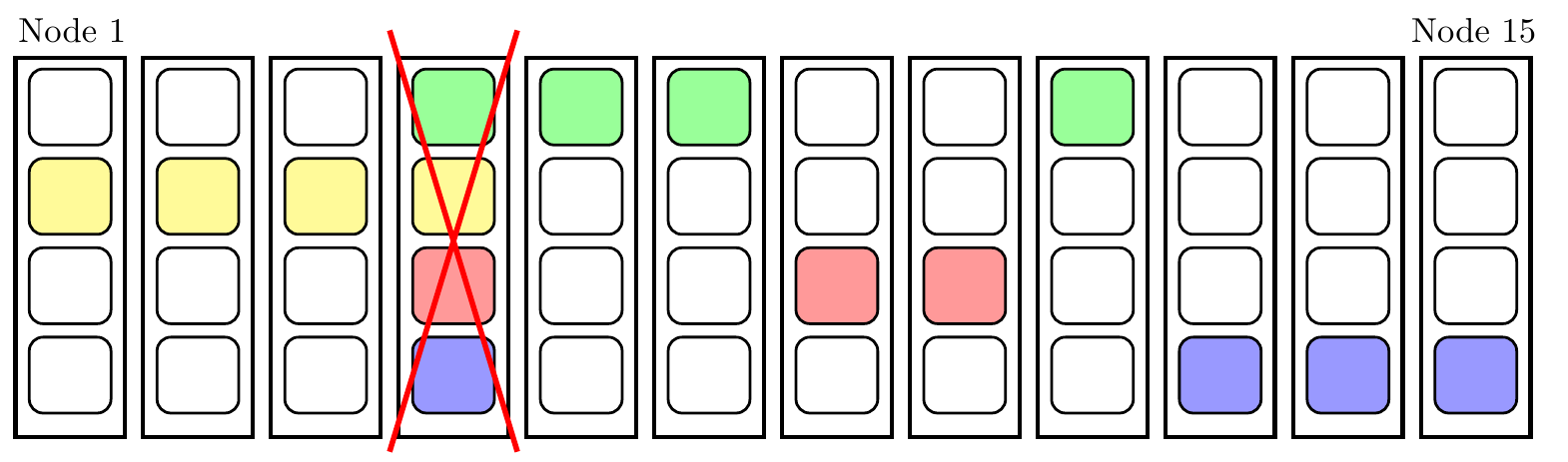}
         \caption{Here, we illustrate $15$ nodes of a large scale storage system, where different files are encoded using exact-repairable bandwidth efficient codes. We assume that the fourth node (from the left) is lost due to a failure. The failed node stored four code blocks corresponding to four distinct codewords (with their code blocks highlighted with different colors). As one can notice that the codewords corresponding the lost code blocks are distributed among all the $15$ nodes and the load of providing the required information for repair of the lost four code blocks gets distributed among all the nodes. Furthermore, the small sub-packetization implies that the process of downloading the required information from these contributing nodes can be very fast as we don't need to download a large amount of information from each of these nodes.}
         \label{fig:repair_balance}
\end{figure*}

\noindent \textbf{Our contributions.}~We present two general approaches to construct MDS codes that have small sub-packetization level while allowing for exact repair of all code blocks with near-optimal repair bandwidth. The constructions obtained using the proposed approaches highlight a trade-off between the sub-packetization level and the repair bandwidth for exact repair. Both approaches crucially utilize the parity-check view of a linear code to obtain the desired MDS codes. We note that throughout this paper we consider the setting with $t = n-1$, i.e., all the remaining code blocks contribute to the exact repair of a single code block.

The first approach that we present aims to construct repair-efficient MDS codes over an extension field by utilizing the Vandermonde style parity-check matrices for MDS codes over the base field. Assuming that the desired sub-packetization level is $\ell$, we take $\F$ to be the degree $\ell$ extension of the base field $\B$. We start with a block parity-check matrix of a simple MDS code (over $\B$) of length $n\ell$ which is obtained by naturally interleaving $\ell$ independent codewords from the given MDS code (over $\B$) with a Vandermonde style parity-check matrix. We then carefully replace some of the zero entries of this parity-check matrix with non-zeros elements from $\B$ and obtain a parity-check matrix of a new MDS code (over $\F$) of length $n$ that has an exact repair mechanism with small repair bandwidth. It follows from the construction that the obtained repair-efficient MDS code is linear over the base field $\B$. 

We demonstrate our first approach by presenting a family of MDS codes that have sub-packetization level $\ell = n - k$ and the repair bandwidth that is strictly less than $2(n-1)$ symbols (over $\B$). Note that this is twice the cut-set bound (cf.~\eqref{eq:cut_set_d}) which takes the value $(n-1)$ symbols (over $\B$) for $\ell = n - k$. {We argue that the sub-packetization level $\ell = \Omega(n - k)$ is the smallest that one can hope for an MDS code with the aforementioned guarantee on its repair bandwidth (See Appendix~\ref{appen:appen_sub}).} We then generalize the ideas used in the construction with $\ell = n - k$ to obtain the MDS codes that have improved repair bandwidth at the cost of increased sub-packetization level. In particular, for an integer $\tau \geq 2$, we obtain a family of MDS codes with sub-packetization level $\ell = (n - k)^{\tau}$ and the repair bandwidth which is at most $(1 + \frac{1}{\tau})$ times the value of the cut-set bound.

Note that the existing constructions of MDS codes that attain the cut-set bound (e.g.~\cite{RSK11, YeB16a}) have an additional property that these codes ensure load balancing during the repair process. In particular, the repair mechanisms of these codes require downloading the same amount of data from each of the $t$ contacted nodes during the repair of a failed node. In the literature, such codes are referred to as {\em minimum storage regenerating (MSR)} codes~\cite{dimakis}. We note that the MDS codes obtained from our first approach do not provide the load balancing feature as some of the nodes contribute higher amount of data during the node repair process. The second approach that we propose in this paper addresses this issue. In this general approach, we start with a short MSR code that has large sub-packetization level as a function of its length. Starting with the parity-check matrix of this code, we apply a simple yet powerful transformation to obtain a parity-check matrix of a longer code. The obtained longer code has a sub-packetization level which scales much favorably with the code length while ensuring a small repair bandwidth for exact repair. Furthermore, the repair mechanism of the longer code ensures load balancing during the repair process. We name the codes obtained through this approach as $\epsilon$-MSR codes: For a given $\epsilon \geq 0$, an $\epsilon$-MSR code downloads at most 
$
(1 + \epsilon)\cdot{\ell}/{(t - k  + 1)}~~\text{symbols (over $\B$)}
$
from each of the $t$ contacted nodes during the repair of a failed node. We highlight that this approach allows us to construct $\epsilon$-MSR codes with constant number of parity nodes while working with the sub-packetization level that scales only {\em logarithmically} with the code length $n$. This amounts to a {\em doubly exponential} saving in terms of the sub-packetization level as compare to the existing MSR codes~\cite{YeB16a, SVK16}. Finally, we address the question of finding the sub-packetization level that is necessary for any $\epsilon$-MSR code with the given code length and rate. Towards this, we obtain an upper bound on the code length $n$ of a {\em linear} $\epsilon$-MSR code as a function of $n - k$ and sub-packetization level $\ell$.

For the MDS codes obtained from our first approach, the exact repair of a code block involves downloading a subset of the symbols from the remaining code blocks. Such repair mechanisms are referred to as the {\em uncoded repair} or {\em repair-by-transfer} in the literature. The repair-by-transfer schemes form a sub-class of all possible linear repair schemes where a contacted node can potentially send symbols (over $\B$) which are linear combinations of all $\ell$ symbols of the code block stored on this node. A repair-by-transfer mechanism is desirable over other complicated repair schemes due to its operational simplicity and the minimal computation requirements at the contacted nodes. Here, we would also like to point out that our second approach can also be utilized to obtain the codes that enable repair-by-transfer by working with those short MSR codes that have repair-by-transfer schemes (cf.~Remark~\ref{rem:RBT_long}).

As for the comparison between the proposed two approaches, for the similar guarantee on the repair bandwidth, the second approach requires slightly higher sub-packetization level (cf. Remark~\ref{rem:GR17}). However, as pointed out above, the second approach gives us $\epsilon$-MSR codes that have the additional feature of load balancing, where all the contacted code blocks provide (approximately) same amount of information during the repair process. Moreover, the second approach allows us to construct explicit codes with the field size $|\B|$ which scales linearly with the code length. In contrast, we manage to obtain explicit construction using the first approach with the field size that is much larger than $n^{(r-1)\ell + 1}$. That said, one of the merits of the first approach over the second approach is the flexibility it provides in terms of various system parameters. The first approach only relies on MDS codes with Vandermonde style parity check matrices which exist for all values of $n$ and $k$. Moreover, the desired repair bandwidth and the associated sub-packetization level can be easily controlled with the design parameter $\tau$ which does not depend on $n$ and $k$. On the other hand, the second approach requires MSR codes and error correcting codes with specific parameters to obtain the $\epsilon$-MSR codes with the desired $n$, $k$ and sub-packetization level, which may not always be available. 

\noindent \textbf{Organization.}
We introduce the necessary background along with a discussion on the related work in Section~\ref{sec:background}. In Section~\ref{sec:main}, we formally define the notion of MDS codes with near-optimal repair bandwidth and summarize the key contributions of this paper. We present the first contribution of this paper in Section~\ref{sec:twiceRB}, where we describe a construction of MDS codes with sub-packetization level equals to the number of parity nodes and repair bandwidth at most twice the value of the cut-set bound. In Section~\ref{sec:RB_t}, we generalize this construction to obtain MDS codes with their repair bandwidth approaching the cut-set bound. In Section~\ref{sec:gen_approach}~and~\ref{sec:long_mds}, we focus on constructing $\epsilon$-MSR codes. In Section~\ref{sec:gen_approach}, we describe a general approach to combine a short MSR code with an error correcting code into a code with both small repair bandwidth and small sub-packetization level. Furthermore, the obtained code also ensures that none of the contacted code blocks contribute disproportionately large number of symbols during a node repair. In Section~\ref{sec:long_mds}, we utilize a code construction of \cite{YeB16a} to construct a short MSR code in our general approach so that the obtained code is an MDS code. This along with the load balancing nature of the repair mechanism of the obtained code establish that it is an $\epsilon$-MSR code. In Section~\ref{sec:converse}, we explore the sub-packetization level that is necessary to realize an $\epsilon$-MSR code. We conclude the paper in Section~\ref{sec:conclusion} where we comment on the constructions of the codes with general values of $t$ (the number of blocks contributing to the repair process) and discuss other directions for future work.

\section{Background and Related Work.}
\label{sec:background}

In this section we formally introduce linear array codes and the related concepts used in this paper. We then describe the repair problem in the context of distributed storage systems and survey the related work. Throughout this paper we use the following notation. For an integer $a >0$, $[a]$ denotes the set $\{1, 2,\ldots, a\}$. Similarly, for two integers $a$ and $b$ such that $a \leq b$, $[a:b]$ represents the set $\{a, a+1,\ldots, b\}$. For two matrices $A$ and $B$, $A \otimes B$ denotes the tensor product of $A$ and $B$.

%%%%%%%%%%%%%%%%%%%%%%%%%%%%%%%%%%%%%%%%%%%%%%%%%%%%%%%%%%%%%%%
%%%%%%%%%%%%%%%%%%%%%%%%%%%%%%%%%%%%%%%%%%%%%%%%%%%
%% Preliminaries
%%%%%%%%%%%%%%%%%%%%%%%%%%%%%%%%%%%%%%%%%%%%%%%%%%%
%%%%%%%%%%%%%%%%%%%%%%%%%%%%%%%%%%%%%%%%%%%%%%%%%%%%%%%%%%%%%%%
\subsection{Linear array codes}
\label{sec:array}
Let $\F$ be a degree $\ell$ extension of the finite field $\B$. We say that a set $\Cc \subseteq \F^{n}$ forms an $(n, M, d_{\min})_{\F}$ code, if we have $|\Cc| = M$ and
$$
d_{\min} = \min_{\cv \neq \cv' \in \Cc}\dist(\cv, \cv'),
$$
where $\dist(\cdot,\cdot)$ denotes the Hamming distance. %Given a subfield $\B$ such that $\F$ is its degree $\ell$ extension, 
Note that each element of $\F$ can be represented as an $\ell$-length vector over $\B$. Therefore, we can express a codeword $\cv = (c_1,\ldots, c_n) \in \Cc \subseteq \F^{n}$ as an $n\ell$-length vector %over $\B$ as follows.
$
\cv = (\cv_1,\ldots, \cv_n) \in \B^{n\ell}.
$
Here, for $i \in [n]$, the {\em code block} $\cv_i = (c_{i, 1},\ldots, c_{i, \ell}) \in \B^{\ell}$ denotes the $\ell$-length vector corresponding to the {\em code symbol} $c_i \in \F$. 

In this equivalent representation, we say that $\Cc$ is a {\em linear array code} if it forms a linear subspace over $\B$. Moreover, we refer to the code as an $[n, \log_{|\B|}M, d_{\min}, \ell]_{\B}$ linear array code. We say that an $[n, \log_{|\B|}M, d_{\min}, \ell]_{\B}$ linear array code is a {\em maximum distance separable} (MDS) code if $\ell$ divides $\log_{|\B|}M$ and $$d_{\min} = n - \big({\log_{|\B|}M}\big)/{\ell} + 1.$$
\begin{remark}
Note that even though we view the codewords of an array code as $n\ell$-length vectors over $\B$, the minimum distance is calculated by viewing each code block as a symbol of  $\F$. Therefore, the minimum distance of the array code belongs to the set of integers $[n]$.
\end{remark}
An $[n, \log_{|\B|}M, d_{\min}, \ell]_{\B}$ linear array code can be defined by an $(n\ell - \log_{|\B|}M) \times n\ell$ full rank matrix $\Hm$ over $\B$ as follows.
\begin{align}
\label{eq:PM}
\Cc = \big\{\cv = (\cv_1,\ldots, \cv_n)~:~\Hm\cdot\cv = 0\big\} \subseteq \B^{n\ell}.
\end{align}
The matrix $\Hm$ is called a {\em parity-check matrix} of $\Cc$. Assuming that $k$ is an integer such that $\log_{|\B|}M = k\ell$, the parity-check matrix $\Hm$ can be viewed as a block matrix
\begin{align}
\Hm & = \left(\begin{array}{cccc}
\Hm_1 & \Hm_2 & \cdots & \Hm_n 
\end{array}\right)  \in \B^{(n-k)\ell \times n\ell}.
\end{align}
For $i \in [n]$, we refer to the $(n-k)\ell \times \ell$ sub-matrix $\Hm_i$ as the {\em thick column} associated with the $i$-th code block in the codewords of $\Cc$. For a set $\Sc = \{i_1, i_2,\ldots, i_{|\Sc|}\} \subseteq [n]$, we define the $(n-k) \ell \times |\Sc|\ell$ matrix $\Hm_{\Sc}$ as follows. 
\begin{align}
\label{eq:ParitySub}
\Hm_{\Sc} = \left( \begin{array}{cccc}
\Hm_{i_1} & \Hm_{i_2} & \cdots & \Hm_{i_{|\Sc|}}
\end{array} \right) \in \B^{(n-k) \ell \times |\Sc|\ell}.
\end{align}
Note that the matrix $\Hm_{\Sc}$ comprises the thick columns with indices in the set $\Sc$. The parity-check matrix $\Hm$ defines an MDS code if for every $\Sc \subseteq [n]$ with $|\Sc| = n-k$, the $(n-k)\ell \times (n-k)\ell$ sub-matrix $\Hm_{\Sc}$ is full rank.

\subsection{Repair problem for MDS codes.}
\label{sec:exact-repair}
Let $\cv = (\cv_1,\ldots, \cv_n) \in \B^{n\ell}$ be a codeword of an MDS code $\Cc$ with $\log_{|\B|}{M} = k\ell$. Recall that we are considering a distributed storage setup where these $n$ code blocks are stored on $n$ distinct storage nodes. For every $i \in [n]$, we are interested in the task of repairing the code block $\cv_i$ by downloading a small amount of data from the remaining nodes storing the code blocks $\big\{\cv_{j}\big\}_{j\neq i}$. Accordingly, with $k \leq t \leq n - 1$ as another system parameter, the repair problem imposes the requirement that for every $i \in [n]$ and $\Rc \subseteq [n]\backslash \{i\}$ with $|\Rc| = t$, we have a collection of functions, 
$$
\big\{h^{(i)}_{j, \Rc}~:\B^{\ell} \rightarrow \B^{\beta_{j,i}}\big\}_{j \in \Rc}
$$
such that $\cv_i$ is a function of the symbols in the set $\big\{h^{(i)}_{j, \Rc}(\cv_j)\big\}_{j \in \Rc}$. This implies that  for every $i \in [n]$, the code block $\cv_i$ can be {\em exactly repaired} (regenerated) by contacting any $t$ out of $n-1$ remaining code blocks in the codeword $\cv$ (say indexed by the set $\Rc \subseteq [n]\backslash \{i\}$) and downloading at most $\sum_{j \in \Rc}\beta_{j, i}$ symbols  of $\B$ from the contacted code blocks, where $\beta_{j,i}$ to denote the number of symbols of $\B$ downloaded from the code block $\cv_j$.

In \cite{dimakis}, Dimakis et al. formally study the repair problem for MDS codes in the setup described above. They introduce {\em repair bandwidth}, the total number of symbols downloaded during the repair process, as a measure to characterize the efficiency of the repair process\footnote{Dimakis et al. consider a broader repair framework, namely {\em functional repair} framework~\cite{dimakis}. Under functional repair framework, $\tilde{\cv}_i \in \B^{\ell}$ which may potentially be different from the code block under repair $\cv_i \in \B^{\ell}$ is an acceptable outcome of the repair process as long as it preserves certain properties of the original codeword. For further details, we refer the reader to \cite{dimakis, DRWS2011}. Here, we note that the lower bounds obtained for the functional repair problem are also applicable to the exact repair problem considered in this paper.}. Subsequently, Dimakis et al. obtain the following cut-set bound on the repair bandwidth of an MDS code~\cite{dimakis}.
\begin{align}
\label{eq:cut-set}
\sum_{j \in \Rc}\beta_{j, i} \geq \left(\frac{t}{t - k + 1}\right)\cdot \ell,~~\forall~\Rc \subseteq [n]\backslash \{i\}~\text{s.t.}~|\Rc| = t.
\end{align}
Note that the cut-set bound (cf.~\eqref{eq:cut-set}) is a decreasing function of $t$. Therefore, by contacting more nodes during the repair process, one potentially needs to download less overall information from the contacted nodes. This makes the setting with $t = n-1$ the most beneficial in terms of minimizing the repair bandwidth.

\subsection{Linear repair schemes for a linear array code}
\label{sec:linearRepair}

Before we provide a brief account of the work on constructing MDS codes that attain the cut-set bound in Section~\ref{sec:prior_work}, let us first introduce various concepts pertaining to linear repair schemes for a linear array code. Since we focus on constructing codes that are linear and employ linear schemes to repair a code block, these concepts are instrumental for the presentation of the main results of this paper. Since all the constructions presented in this paper work with $t = n-1$, in what follows, we restrict ourselves to this setting. 

Let $n$ nodes in the distributed storage setup store $n$ distinct code blocks of a codeword $\cv = (\cv_1,\ldots, \cv_n)$ belonging to a linear array code $\Cc$ defined by the parity-check matrix $\Hm$ (cf.~Section~\ref{sec:array}). {For every $i \in [n]$, a linear repair scheme performs the task of repairing the code block $\cv_i$ with the help of remaining $t = n-1$ code blocks $\big\{\cv_{j}\big\}_{j\neq i}$ by employing linear operation over the base field $\B$. We summarize the description of a linear repair scheme and the associated repair bandwidth in the following statement. 

\begin{proposition}
\label{prop:repair}
Let $S_i \in \B^{\ell \times (n - k)\ell}$ be a matrix such that the following two conditions hold. 
\begin{align}
\label{eq:Repair_cond1}
\rank\left(S_i\Hm_i \right) = \ell
\end{align}
and
\begin{align}
\label{eq:Repair_cond2}
\sum_{j \in [n]\backslash\{i\}}\rank\left(S_i\Hm_j\right) \leq \gamma,
\end{align}
where all the ranks are calculated over the base field $\B$. Then, the code block $\cv_i$ can be repaired by downloading at most $\gamma$ symbols of $\B$ from the remaining $n - 1$ nodes. In particular, this requires downloading $\rank\left(S_i\Hm_j\right)$ symbols of $\B$ from the node storing the code block $\cv_j$. Furthermore, $\{S_i\}_{i \in [n]}$ are referred to as repair matrices. 
\end{proposition}

\begin{proof}
Given the matrix $S_i$, one can regenerate the code block $\cv_i$ by downloading at most $\gamma$ symbols (over $\B$) from the remaining $n - 1$ nodes as follows.
\begin{itemize}
\item Note that the codeword $\cv  = (\cv_1, \cv_2,\ldots, \cv_n)$ satisfies the following. 
\begin{align}
\label{eq:Repair_parity1}
\Hm \cv = \Hm_1\cv_1  + \cdots + \Hm_n\cv_n = 0
\end{align}
\item By multiplying \eqref{eq:Repair_parity1} from left by the matrix $S_i$ which satisfies \eqref{eq:Repair_cond1} and \eqref{eq:Repair_cond2}, we obtain
\begin{align}
\label{eq:Repair_parity2}
S_i \Hm_i\cv_i  = - \sum_{j \in [n]\backslash\{i\}}S_i \Hm_j\cv_j
\end{align}
\end{itemize} 
Note that in order to evaluate the right hand side of \eqref{eq:Repair_parity2}, for $j \in [n]\backslash\{i\}$, we need to download at most
$
\rank\left(S_i\Hm_j\right)
$
symbols of $\B$ from the node storing the code block $\cv_j$. It follows from \eqref{eq:Repair_cond2} that we download at most $\gamma$ symbols from the code blocks $\{\cv_1,\ldots, \cv_{i-1}, \cv_{i+1},\ldots, \cv_n\}$. Once we know the right hand side of \eqref{eq:Repair_parity2}, we can solve for $\cv_i$ as it follows from \eqref{eq:Repair_cond1} that the matrix $S_i\Hm_i$ is full rank.
\end{proof}

\begin{remark}
\label{rem:zeroRank}
In the repair mechanism summarized in Proposition~\ref{prop:repair}, if $\rank\left(S_i\Hm_j \right)  = 0$ i.e., this is a zero matrix, then evidently the node storing $\cv_j$ does not participate in the repair process. 
\end{remark}

\begin{remark}
\label{rem:RBT}
Let $\{\ev_i\}_{i \in [\ell]}$ be $\ell$ standard basis vectors of $\B^{\ell}$, where all but $i$-th coordinate of the vector $\ev_i$ are zero. 
{We say that the code block $\cv_i$ is repaired by transfer if for all $j\neq i$, there exists a subset $\Vc_j\subseteq [\ell]$ such that}
\begin{align}
\label{eq:RBT}
{\rm rowspace}(S_i\Hm_j)= {\rm span}\{\ev^{T}_s : s \in \Vc_j\}.
\end{align}
Here, ${\rm row space}(S_i\Hm_j)$ denotes the subspace spanned by the rows of the matrix $S_i\Hm_j$. Given that \eqref{eq:RBT} holds, it is sufficient to simply transfer $\rank\left(S_i\Hm_j\right)$ symbols of $\B$ from the code block $\cv_j$ (without any local computation at the $j$-th node) in order to repair the code block $\cv_i$, which justifies the term {\em repair-by-transfer}.
\end{remark}

\subsection{Prior work on constructing repair-efficient codes}
\label{sec:prior_work}

The problem of constructing MSR codes, i.e., MDS codes that attain the cut-set bound (cf.~\eqref{eq:cut-set}), has been explored by many researchers. In \cite{RSK11}, Rashmi et al. present an explicit construction for MSR codes. This construction works with the sub-packetization level $\ell = t - k + 1\leq n-k$. However, this small sub-packetization is achieved at the expense of low rate which is bounded as $\frac{k}{n} \leq \frac{1}{2} + \frac{1}{2n}$. Towards constructing high-rate MSR codes, Cadambe et al.~\cite{CJMRS13} show the existence of such codes when sub-packetization level approaches infinity. Motivated by this result, the problem of designing high-rate MSR codes with finite sub-packetization level is explored in \cite{PapDimCad13, zigzag13, SAK15, WTB12, Cadambe_poly, RKV16a, GoparajuFV16, YeB16a, YeB16b, SVK16} and references therein. The constructions presented in \cite{PapDimCad13, zigzag13, zigzag_allerton11}  work with the sub-packetization level $\ell$ which is exponential in $k$. For $t = n-1$ and all values of $r = n - k$, Sasidharan et al.~\cite{SAK15} construct MSR codes with the sub-packetization level  $\ell = (n - k)^{\ceilb{\frac{n}{n-k}}}$. {The constructions with $t = n - 1$ and the similar sub-packetization levels that enable exact-repair of only $k$ systematic nodes are also presented in \cite{WTB12, Cadambe_poly}.} The construction of \cite{SAK15} is generalized for all possible values of $k \leq t \leq n - 1$ with the sub-packetization level $\ell = (t - k + 1)^{\ceilb{\frac{n}{t - k + 1}}}$ in \cite{RKV16a}. 

The MSR codes presented in \cite{PapDimCad13, zigzag13, zigzag_allerton11, WTB12, Cadambe_poly} are obtained by designing suitable generator matrices for these codes. On the other hand, \cite{SAK15, RKV16a} design the proposed codes by constructing parity-check matrices with certain combinatorial structures. {We note that in most of these constructions, certain elements in the generator/parity-check matrices are not explicitly specified. These papers argue the existence of good choices for these elements provided that the field size is large enough.} Recently, Ye and Barg~\cite{YeB16b} have presented a fully explicit construction for MSR codes with $t = n-1$ and the sub-packetization level $\ell = (n - k)^{\ceilb{\frac{n}{n - k}}}$ by designing the associated parity-check matrices. We note that the construction from \cite{YeB16b} also works for general values of $k \leq t \leq n - 1$ with suitably modified sub-packetization levels similar to the sub-packetization levels used in \cite{RKV16a}. A similar construction is also presented in an independent work by Sasidharan et al.~\cite{SVK16}. 

\bgroup
\def\arraystretch{1.55}
\begin{table*}[t!]
\footnotesize
\centering
\begin{tabular}{c|c|c|c|c|}
  \hline \hline
  Code construction & Sub-packetization level & Repair bandwidth & Repair by transfer & Code rate \\
  \hline
  Rashmi et al., 2011 \cite{RSK11} & $\ell = n - k$ & $\big(\frac{n-1}{n-k}\big)\cdot \ell$ & No & $0 < \frac{k}{n} \leq \frac{1}{2} + \frac{1}{n}$ \\
  \hline
  Ye and Barg, 2016 \cite{YeB16b} & $\ell = (n - k)^{\ceilb{\frac{n}{n-k}}}$ & $\big(\frac{n-1}{n-k}\big)\cdot \ell$ & Yes & $0 < \frac{k}{n} < 1$ \\
  \hline
  \pbox{20cm}{~~~~~~~~~~This paper \\ (design parameter $\tau \geq 1$)} & $\ell = (n - k)^{\tau}$ & $\leq (1 + \frac{1}{\tau})\cdot\big(\frac{n-1}{n-k}\big)\cdot \ell$ & Yes & $0 < \frac{k}{n} < 1$ \\
\hline
\end{tabular}
\caption{Comparison of the code construction proposed in Section~\ref{sec:RB_t} with the state of the art constructions of MSR codes. We focus only on the setting with $t = n - 1$.}
\label{tab:comparison}
\end{table*}
\egroup

Some converse results on the sub-packetization level that is necessary for an MSR code are presented in \cite{GTC14, TWB14}. For $t = n - 1$, Goparaju et al.~\cite{GTC14} show that an MSR code that employs linear repair schemes satisfies the following bound on its sub-packetization level. 
\begin{align}
\label{eq:gtc_bound}
k \leq 2(\log_{2}\ell)\big(\log_{\frac{n-k}{n-k-1}}\ell + 1\big) + 1.
\end{align}
Note that, for $n - k = \Theta(1)$, the bound in \eqref{eq:gtc_bound} implies that $\ell = \Omega\big(\exp(\sqrt{k})\big)$. Thus, for the setting with constant number of parity nodes, an MSR code necessarily has a very large sub-packetization level. On the other hand, Tamo et al.~\cite{TWB14} show that the sub-packetization level of an MSR code which enables exact repair using repair-by-transfer schemes is bounded as
\begin{align}
\label{eq:twb_bound}
\ell \geq (n - k)^{\frac{k}{n - k}}.
\end{align}
Note that repair-by-transfer schemes constitute a sub-class of all possible linear repair schemes. In light of the bound in \eqref{eq:twb_bound}, the MSR codes obtained in \cite{SAK15, YeB16b, SVK16} enable repair-by-transfer mechanisms with near-optimal sub-packetization level. However, this sub-packetization level can be prohibitively large for some storage systems, especially when the code has high rate or equivalently has small value of $r = n  - k$. In particular, for constant number of parities, i.e., $r=O(1)$, the sub-packetization level scales exponentially with $n$. This motivates us to explore the question of designing MDS codes that work with small sub-packetization level and provide repair mechanisms without incurring much degradation in terms of the repair bandwidth. In Table~\ref{tab:comparison}, we compare one of our proposed constructions with the previously known constructions.

The problem of constructing exact-repairable MDS codes with small repair bandwidth and small sub-packetization level has been previously addressed in \cite{zigzag13, RSR13}. We note that our first construction (cf.~Section~\ref{sec:twiceRB} and \ref{sec:RB_t}) shares some similarities with the constructions presented in \cite{zigzag13, RSR13} as these constructions are obtained by introducing coupling among multiple independent codes as well. However, we work with the parity-check matrix view (as opposed to the generator matrix view considered in \cite{zigzag13, RSR13}) which ensures identical repair guarantees for all code blocks without distinguishing between systematic and parity code blocks. In a parallel and independent work~\cite{KGO16, KGJO16}, the authors also address the problem of constructing MDS codes with small sub-packetization and near-optimal repair bandwidth. However, they deal with the repair of only systematic nodes. Moreover, their construction is not fully explicit.\\

\noindent \textbf{Exact repair of known codes with small repair bandwidth.} The problem of devising exact repair mechanism with small repair bandwidth for known MDS codes has been studied in \cite{WDB10, SPDG14, XCode14, GW15, YeBargRS, DM17, TYB17}. In particular, \cite{SPDG14, GW15, YeBargRS} consider the exact repair problem for the well-known Reed-Solomon codes. In \cite{GW15}, Guruswami and Wootters present a framework to design linear schemes to repair RS codes and more generally scalar MDS codes, which are linear over $\F$. They further characterize optimal repair bandwidth for RS codes in certain regimes of system parameters. In \cite{YeBargRS}, utilizing the framework from \cite{GW15}, Ye and Barg show that it is possible to construct RS codes with asymptotically optimal repair bandwidth (as the code length $n$ tends to $\infty$). The repair scheme in \cite{YeBargRS} works with a sub-packetization level of $(n-k)^n$ over a base field. Recently, Tamo et al. construct RS codes that meet the cut-set bound in \cite{TYB17}. Their construction requires the sub-packetization level of $\exp(1+o(1)n\log n)$. Furthermore, they show that such a  super-exponential scaling of the sub-packetization level with the code length is necessary for linear repair schemes of scalar MDS codes to attain the cut-set bound.\\

\noindent \textbf{Locally repairable codes.}~Another line of work in distributed storage focuses on {\em locality}, the number of the code blocks contacted during the repair of a single code block,  as a metric to characterize the efficiency of the repair process. The bounds on the failure tolerance of locally repairable codes, the codes with small locality, have been obtained in \cite{Gopalan12, PapDim12, KPLK12, RKSV12} and references therein. Furthermore, the constructions of locally repairable codes that are optimal with respect to these bounds are presented in \cite{Gopalan12, PapDim12, KPLK12, RKSV12, BlaumHH13, TamoBarg14, Gopalan14}. Locally repairable codes that also minimize the repair bandwidth for repair of a code block are considered in \cite{KPLK12, RKSV12}. Here we note that the locally repairable codes are not MDS codes, and thus have extra storage overhead.

%%%%%%%%%%%%%%%%%%%%%%%%%%%%%%%%%%%%%%%%%%%%%%%%%%%%%%%%%%%%%%%
%%%%%%%%%%%%%%%%%%%%%%%%%%%%%%%%%%%%%%%%%%%%%%%%%%%
%% Code Constructions
%%%%%%%%%%%%%%%%%%%%%%%%%%%%%%%%%%%%%%%%%%%%%%%%%%%
%%%%%%%%%%%%%%%%%%%%%%%%%%%%%%%%%%%%%%%%%%%%%%%%%%%%%%%%%%%%%%%

\section{MDS Codes with Near-optimal Repair Bandwidth.}
\label{sec:main}
This paper aims to construct MDS codes with small sub-packetization and near optimal repair bandwidth, i.e., incurring a small (multiplicative) loss as compared to the repair bandwidth specified by the cut-set bound (cf.~\eqref{eq:cut-set}). Towards this, we first introduce the following notion of near-optimal repair bandwidth for MDS codes. 
\begin{definition}
\label{def:ApproxMSR}
{\rm Let $\Cc$ be an $[n, k\ell, d_{\min} = n - k + 1, \ell]_{\B}$ MDS code. We call $\Cc$ to be an $(a, \ell, t)$-exact-repairable MDS code if for every $i \in [n]$ and $\cv = (\cv_1, \cv_2,\ldots, \cv_n) \in \Cc$, we can perform exact repair of the code block $\cv_i$ by contacting $t$ other code blocks and downloading at most 
$
a \big(\frac{t}{t - k + 1}\big)\cdot \ell
$
symbols of $\B$ from the contacted code blocks.}
\end{definition}

\begin{remark}
\label{rem:near_opt1}
It follows from the bound in \eqref{eq:cut-set} that for any MDS code we must have $a \geq 1$. Note that $(a = 1, \ell, t)$-exact-repairable MDS codes correspond to MSR codes. Moreover, we say that an MDS code has near-optimal repair bandwidth if it is an $(a, \ell, t)$-exact-repairable MDS code for a small constant $a$.
\end{remark}

In this paper, we present explicit constructions of $(a, \ell, t)$-exact-repairable MDS codes that simultaneously ensure small values for both $a$ and $\ell$.  Furthermore, we focus on the setting with $t = n - 1$, i.e., all the remaining $n-1$ code blocks are contacted to repair a single code block. The following result summarizes the parameters of $(a, \ell, t = n-1)$-exact-repairable MDS codes obtained in this paper.

\begin{theorem}
\label{thm:main_tau}
For an integer $1 \leq \tau \leq \ceilbb{{n}/{(n-k)}} - 1$, Construction~\ref{subsec:construction_t} gives $\big(a = 1 + {1}/{\tau}, \ell = (n-k)^{\tau}, t = n-1\big)$-exact-repairable MDS codes. Moreover, the obtained codes allow for repair-by-transfer schemes.
\end{theorem}

We present our construction for $\tau =1$, which gives $(a = 2, \ell = n-k, t = n-1)$-exact-repairable MDS codes in Section~\ref{sec:twiceRB}. This construction conveys the main ideas behind our approach and establishes Theorem~\ref{thm:main_tau} for $\tau = 1$. The general construction which establishes Theorem~\ref{thm:main_tau} for all values of $\tau$ is presented in Section~\ref{sec:RB_t}.

\begin{remark}
\label{rem:t_opt}
We note that for a given value of $\tau$, $\big(1 + {1}/{\tau}\big)$ only serves as a clean upper bound on the repair bandwidth of the designed codes. Specifically, if we substitute $\tau = \ceilbb{{n}/{(n-k)}}$ in the general construction (cf.~Section~\ref{sec:RB_t}), we obtain $\big(a = 1, \ell = (n - k)^{\ceilbb{{n}/{(n-k)}}}, t = n-1\big)$-exact-repairable MDS codes, which are MSR codes (cf.~Remark~\ref{rem:near_opt1}). In fact, in this case our construction specializes to the construction from \cite{SAK15}.
\end{remark}

As defined, an $(a, \ell, t)$-exact-repairable MDS code ensures that the total amount of information downloaded from the $t$ code blocks participating in a node repair is comparable to the cut-set bound. However, it is possible that some of these $t$ code blocks may have to contribute significantly large amount of information as compared to the other participating code blocks. Depending on the underlying system architecture, this may not be desirable in some settings. This motivates us to explore a sub-family of $(a, \ell, t)$-exact-repairable MDS codes which also ensures load-balancing among all the contacted blocks during the node repair. In particular, we require that none of the contacted code blocks have to contribute disproportionately large amount of information as compared to the other participating code blocks. 

We note that the MSR codes~\cite{dimakis} ensure that for every $i \in [n]$, it is possible to repair the $i$-th code block by downloading exactly $\frac{\ell}{t - k + 1}$ symbols of $\B$ from each of the $t$ (out of $n-1$) intact nodes. Motivated by this load-balancing property, we now formally define an interesting sub-family of $(a, \ell, t)$-exact-repairable codes. Since the codes from this sub-family have the desirable load-balancing property, we refer to these codes as $\epsilon$-MSR codes.

\begin{definition}{\bf ($\epsilon$-MSR code):~}
\label{def:eMSR}
Let $\epsilon>0$ and $\Cc$ be an $[n, k\ell, d_{\min} = n - k + 1, \ell]_{\B}$ MDS code. We say that $\Cc$ is an $(n, k, t, \ell)_{\B}$ $\epsilon$-MSR code (or simply an $\epsilon$-MSR code) if, for every $i \in [n]$, there is a repair scheme to repair the $i$-th code block $\cv_i$ with 
\begin{align*}
\beta_{j, i} \leq (1 + \epsilon)\cdot \frac{\ell}{t - k + 1}~\text{symbols (over $\B$)}~~\forall~j \in \Rc \subseteq [n]\backslash\{i\}~\text{s.t.}~|\Rc| = t.
\end{align*}
Here, $\beta_{j, i}$ denotes the number of symbols that the code block $\cv_j$ contributes during the repair of the code block $\cv_i$.
\end{definition}
\begin{remark}
Note that an $(n, k, t, \ell)_{\B}$ $\epsilon$-MSR code is also an $(a = 1 + \epsilon, \ell, t)$-exact-repairable MDS code (cf.~Definition~\ref{def:ApproxMSR}). Moreover, one can also notice that an $(n, k, t, \ell)_{\B}$ $\epsilon$-MSR code with $\epsilon = 0$ is an MSR code~\cite{dimakis}. Therefore, we simply refer to such a code as an $(n, k, t, \ell)_{\B}$ MSR code.
\end{remark}

In this paper, we present a general approach for constructing $\epsilon$-MSR codes with small sub-packetization level. The approach utilizes two codes: 1) an MSR code with large sub-packetization level and 2) a code with large enough minimum distance. We describe this general approach in Section~\ref{sec:gen_approach}. The exact dependence among the parameters of the obtained $\epsilon$-MSR code depends on the particular choice of two codes utilized to employ our general approach. In Section~\ref{sec:long_mds}, we work with the specific MSR codes of~\cite{YeB16a} to present explicit instances of $\epsilon$-MSR codes. Additionally, using the codes that operate at the GV curve~\cite{MacSlo} as the codes with large minimum distance, we establish the following result.

\begin{theorem}
\label{thm:param_main}
For any $\epsilon>0$ and a positive integer $r$, there exists a constant $s=s(r,\epsilon)>0$ such that for infinite values of $\ell$ there exists an $\big(n = \Omega(\exp(s\ell)), k = n - r, t = n - 1, \ell\big)_{\B}$ $\epsilon$-MSR code. Furthermore, the required field size $|\B|$ scales as $O(n)$. 
\end{theorem}

It follows from Theorem~\ref{thm:param_main} that for the setting with constant number of parity code blocks, i.e., $r = \Theta(1)$, there exist $\epsilon$-MSR codes with the sub-packetization level which scales as $O(\log n)$. Finally, we focus on the issue of characterizing the length of the longest $\epsilon$-MSR code for the given rate and sub-packetization level. Alternatively, we explore the minimum sub-packetization level that is necessary to realize an $\epsilon$-MSR code with the given code length $n$ and number of parity code blocks $r$. In Section~\ref{sec:converse}, we establish the following result in this direction.
\begin{theorem}
\label{thm:necessary}
In an $(n,k = n - r, t = n - 1, \ell)_{\B}$ linear $\epsilon-$MSR code, the number of nodes $n$ is upper bounded by
$(r\ell)^{\frac{\ell}{r}(1+\epsilon)+1}$.  
\end{theorem}

%%%%%%%%%%%%%%%%%%%%%%%%%%%%%%%%%%%%%%%%%%%%%%%%%%%%%%%%%%%%%%%
%%%%%%%%%%%%%%%%%%%%%%%%%%%%%%%%%%%%%%%%%%%%%%%%%%%
%% Code construction
%%%%%%%%%%%%%%%%%%%%%%%%%%%%%%%%%%%%%%%%%%%%%%%%%%%
%%%%%%%%%%%%%%%%%%%%%%%%%%%%%%%%%%%%%%%%%%%%%%%%%%%%%%%%%%%%%%%

\section{Construction of $(2, n-k, n - 1)$-exact-repairable MDS Codes.}
\label{sec:twiceRB}

In this section, we present a construction of exact-repairable MDS codes for all values of $n$ and $k$. These codes have sub-packetization level $\ell = n-k$ and require $t = n-1$ code blocks during the repair process. Furthermore, the repair bandwidth of these codes is at most
$
2\big(\frac{n-1}{n - k}\big)\cdot \ell,
$
which is twice the cut-set bound (cf.~\eqref{eq:cut-set}). We first describe the construction. We then illustrate the repair-by-transfer scheme for the obtained codes in Section~\ref{subsec:repair2}. We argue the MDS property for the construction in Section~\ref{subsec:mdsness2}.

\begin{construction}
\label{subsec:construction2}
Let $r = n - k$. For ease of exposition, we assume that $r | n$ and $n = sr$. We partition the $n$ code blocks in $r = n-k$ groups of size $s$ each\footnote{For a setting where $r\nmid n$, we can partition the $n$ code blocks in $r = n - k$ groups, $n~({\rm mod}~r)$ groups with $\ceilbb{\frac{n}{r}}$ code blocks and the remaining groups with $\floorbb{\frac{n}{r}}$ code blocks. The rest of the construction can be easily modified to work in this case as well.}. This partitioning allows us to index each code block by a tuple $(u, v)$ where $u \in [r] = [n-k]$ and $v \in [s]$. In particular, for $i \in [n]$ the associated tuple $(u, v)$ satisfies
$
i = (u-1)s + v.
$
With this notation in place, for $(u, v) \in [r] \times [s]$, we denote the $\big((u-1)s + v\big)$-th code block as 
\begin{align*}
\cv_{(u-1)s + v} &= \cv_{(u, v)} \\
&\overset{(i)}{=} \big(c(1; (u, v)),\ldots, c(r; (u, v)) \big) \in \B^{r},
\end{align*}
where, for $x \in [r]$, $c(x; (u, v))$ denotes the $x$-th symbol (over $\B$) of the $\big((u-1)s + v\big)$-th code block. Note that $(i)$ follows from the fact that we have $\ell = r$.
\begin{figure*}[htbp]
\small
\begin{align*}
\Pm &= \left(\begin{array}{ccc|ccc?ccc|ccc ? ccc|ccc}
1 & 0 & 0 &1 & 0 & 0 &1 & 0 & 0 &1 & 0 & 0 &1 & 0 & 0 &1 & 0 & 0  \\ 
0 & 1 & 0 &0 & 1 & 0 &0 & 1 & 0 &0 & 1 & 0 &0 & 1 & 0 &0 & 1 & 0  \\ 
0 & 0 & 1 &0 & 0 & 1&0 & 0 & 1 &0 & 0 & 1 &0 & 0 & 1 &0 & 0 & 1  \\  \hline
\lambda_1 & {\psi} & 0 & \lambda_2 & {\psi} & 0 &\lambda_3 & 0 & 0 &\lambda_4 & 0 & 0 &\lambda_5 & 0 & 0 &\lambda_6 & 0 & 0  \\ 
0 & \lambda_1 & 0 & 0 & \lambda_2 & 0 & 0 &\lambda_3 & {\psi} & 0 &\lambda_4 & {\psi} & 0 &\lambda_5 & 0 & 0 &\lambda_6 & 0  \\  
0 & 0 & \lambda_1 & 0 & 0 & \lambda_2 & 0 & 0 &\lambda_3 & 0 & 0 &\lambda_4 & {\psi} & 0 &\lambda_5 & {\psi} & 0 &\lambda_6 \\  \hline
\lambda^2_1 & 0 & {\psi} & \lambda^2_2 & 0 & {\psi} &\lambda^2_3 & 0 & 0 &\lambda^2_4 & 0 & 0 &\lambda^2_5 & 0 & 0 &\lambda^2_6 & 0 & 0  \\ 
0& \lambda^2_1 & 0 & 0 & \lambda^2_2 & 0 & {\psi} &\lambda^2_3 & 0 & {\psi} &\lambda^2_4 & 0 & 0 &\lambda^2_5 & 0 & 0 &\lambda^2_6 & 0 \\ 
0 & 0 & \lambda^2_1 & 0 & 0 & \lambda^2_2 & 0 & 0 &\lambda^2_3 & 0 & 0 &\lambda^2_4 & 0 & {\psi} &\lambda^2_5 & 0 & {\psi} &\lambda^2_6  \\ 
\end{array} \right).
\end{align*}
\caption{The parity-check matrix obtained from the system parameter in Example~\ref{ex:ex1}.}
\label{fig:PP}
\end{figure*}
In order to construct an $[n, k\ell, d_{\min} = n- k + 1, \ell = r]_{\B}$ MDS code $\Cc$, we specify an $r\ell \times n\ell$ (or $r^2 \times nr$ for our choice of $\ell$) parity-check matrix $\Pm$ for the code $\Cc$. Let $\Lf$ be finite field of size at least $n + 1$ and $ \{\lambda_i\}_{i \in [n]}$ be $n$ distinct non-zero elements of $\Lf$. We take $\B$ to be an extension field of $\Lf$ such that the following two conditions hold.
\begin{itemize}
\item $\B$ is a simple extension of $\Lf$ which is generated by an element $\psi \in \B$, i.e., $\B = \Lf(\psi)$.
\item The degree of extension $[\B : \Lf]$ is at least $(r  - 1)\ell + 1$. 
\end{itemize}
We classify the linear constraints defined by the parity-check matrix $\Pm$ into two types.
\begin{itemize}
\item {\bf {\rm Type~I} constraints:}~We have $r$ {\rm Type~I} constraints which are defined by the first $r$ rows of the matrix $\Pm$. For every $x \in [r]$, we have
\begin{align}
\label{eq:type1}
\sum_{(u,v) \in [r]\times [s]}c(x; (u,v)) = 0.
\end{align}
In Example~\ref{ex:ex1} below, the {\rm Type~I} constraints correspond to the identity blocks of the matrix $\Pm$ (cf.~Figure~\ref{fig:PP}).
\item {\bf {\rm Type~II} constraints:}~We have $(r-1)\ell = (r-1)r$ {\rm Type~II} constraints which are defined as follows. For every $p \in \{1,\ldots, r - 1\}$ and $x \in [r]$, we have 
\begin{align}
\label{eq:type2}
\underbrace{\sum_{(u, v) \in [r]\times [s]}\lambda^p_{(u-1)s + v} \cdot c(x; (u,v))}_{\text{(a)}} +\underbrace{\sum_{v \in [s]}\psi \cdot c(\overline{x + p}; (x, v))}_{\text{(b)}} &= 0,
\end{align}
where for a strictly positive integer $e$, the quantity $\overline{e}$ is defined as follows.
\begin{align}
\label{eq:mod_def}
\overline{e} = \begin{cases}
r &\mbox{if}~e~({\rm mod}~r) = 0, \\
e~({\rm mod}~r) &\mbox{otherwise}.
\end{cases}
\end{align}
We can partition the {\rm Type~II} constraints (cf.~\eqref{eq:type2}) into $(r - 1)$ groups (each group containing $\ell = r$ linear constraints) according to the value of $p \in \{1,\ldots, r-1\}$. In particular, $r$ constraints associated with the same value of $p$ constitute those $r$ rows of the parity-check matrix $\Pm$ which are indexed by the set $\{pr + 1,\ldots, (p+1)r\}$. (See the non-identity blocks of the matrix $\Pm$ in Figure~\ref{fig:PP}.)
\end{itemize}
\end{construction}

\begin{example}
\label{ex:ex1}
We illustrate the construction with an example. Assume that $n = 6$ and $k = 3$, i.e., $n - k = r = 3$. For these values of the system parameters, our $9 \times 18$ parity-check matrix takes the form illustrated in Figure~\ref{fig:PP}. 

The matrix $\Pm$ can be viewed as the perturbation of the block matrix  $\Hm$ which is obtained by replacing all $\psi$ entries in $\Pm$ with zeros. In particular, we can rewrite the matrix $\Pm$ as 
$$
\Pm = \Hm + \Em^{\psi},
$$
where $\Em^{\psi}$ denotes the $9 \times 18$ matrix which contains all the $\psi$ entries in $\Pm$ (cf.~Figure~\ref{fig:PP}) as its only non-zero entries. (See Figure~\ref{fig:PHE}.) Note that the block matrix $\Hm$ (with diagonal blocks) is a parity-check matrix of an $[n = 6, k\ell = 9, d_{\min} = 4, \ell = 3]_{\B}$ MDS code. Here, we also point out that the matrix $\Hm$ is defined by {\rm Type~I} constraints  (cf.~\eqref{eq:type1}) and the part $(a)$ of the {\rm Type~II} constraints (cf.~\eqref{eq:type2}). Similarly, the perturbation matrix $\Em^{\psi}$ is defined by the part $(b)$ of the {\rm Type~II} constraints (cf.~\eqref{eq:type2}).
\end{example}

\begin{figure*}[htbp]
\small
\begin{align}
\Hm &= \left(\begin{array}{ccc|ccc|ccc|ccc|ccc|ccc}
1 & 0 & 0 &1 & 0 & 0 &1 & 0 & 0 &1 & 0 & 0 &1 & 0 & 0 &1 & 0 & 0  \\ 
0 & 1 & 0 &0 & 1 & 0 &0 & 1 & 0 &0 & 1 & 0 &0 & 1 & 0 &0 & 1 & 0  \\ 
0 & 0 & 1 &0 & 0 & 1&0 & 0 & 1 &0 & 0 & 1 &0 & 0 & 1 &0 & 0 & 1  \\  \hline
\lambda_1 & 0 & 0 & \lambda_2 & 0 & 0 &\lambda_3 & 0 & 0 &\lambda_4 & 0 & 0 &\lambda_5 & 0 & 0 &\lambda_6 & 0 & 0  \\ 
0 & \lambda_1 & 0 & 0 & \lambda_2 & 0 & 0 &\lambda_3 & 0 & 0 &\lambda_4 & 0 & 0 &\lambda_5 & 0 & 0 &\lambda_6 & 0  \\  
0 & 0 & \lambda_1 & 0 & 0 & \lambda_2 & 0 & 0 &\lambda_3 & 0 & 0 &\lambda_4 & 0 & 0 &\lambda_5 & 0 & 0 &\lambda_6 \\  \hline
\lambda^2_1 & 0 & 0 & \lambda^2_2 & 0 & 0 &\lambda^2_3 & 0 & 0 &\lambda^2_4 & 0 & 0 &\lambda^2_5 & 0 & 0 &\lambda^2_6 & 0 & 0  \\ 
0& \lambda^2_1 & 0 & 0 & \lambda^2_2 & 0 & 0 &\lambda^2_3 & 0 & 0 &\lambda^2_4 & 0 & 0 &\lambda^2_5 & 0 & 0 &\lambda^2_6 & 0 \\ 
0 & 0 & \lambda^2_1 & 0 & 0 & \lambda^2_2 & 0 & 0 &\lambda^2_3 & 0 & 0 &\lambda^2_4 & 0 & 0 &\lambda^2_5 & 0 & 0 &\lambda^2_6  \\
\end{array} \right)  \nonumber
\end{align} 
\begin{align}
\Em^{\psi} &= \left(\begin{array}{ccc|ccc|ccc|ccc|ccc|ccc}
0 & 0 & 0 &0 & 0 & 0 & 0 & 0 & 0 & 0 & 0 & 0 & 0 & 0 & 0 &0 & 0 & 0  \\ 
0 & 0 & 0 &0 & 0 & 0 &0 & 0 & 0 &0 & 0 & 0 &0 & 0 & 0 &0 & 0 & 0  \\ 
0 & 0 & 0 &0 & 0 & 0&0 & 0 & 0 &0 & 0 & 0 &0 & 0 & 0 &0 & 0 & 0  \\  \hline
0 & {\psi} & 0 & 0 & {\psi} & 0 & 0 & 0 & 0 & 0 & 0 & 0 &0 & 0 & 0 & 0 & 0 & 0  \\ 
0 & 0 & 0 & 0 & 0 & 0 & 0 & 0 & {\psi} & 0 & 0 & {\psi} & 0 & 0 & 0 & 0 & 0 & 0  \\  
0 & 0 & 0 & 0 & 0 & 0 & 0 & 0 & 0 & 0 & 0 & 0 & {\psi} & 0 & 0 & {\psi} & 0 & 0 \\  \hline
0 & 0 & {\psi} & 0 & 0 & {\psi} & 0 & 0 & 0 & 0 & 0 & 0 & 0 & 0 & 0 & 0 & 0 & 0  \\ 
0& 0 & 0 & 0 & 0 & 0 & {\psi} & 0 & 0 & {\psi} & 0 & 0 & 0 & 0 & 0 & 0 & 0 & 0 \\ 
0 & 0 & 0 & 0 & 0 & 0 & 0 & 0 & 0 & 0 & 0 & 0 & 0 & {\psi} & 0 & 0 & {\psi} & 0 \\ 
\end{array} \right)  \nonumber
\end{align}
\caption{Illustration of matrices $\Hm$ and $\Em^{\psi}$ in Example~\ref{ex:ex1}.} 
\label{fig:PHE}
\end{figure*}
%%%%%%%%%%%%%%%%%%%%%%%%%%%%%%%%%%%%%%%%%%%%%%%%%%%%%%
%%%%%%%%%%%%%%%%%%%%%%%%%%%%%%%%%%%%%%%%%%%%%%%%%%%%%%%%%%%%%%%%%%%%%

\subsection{Exact repair of a code block.}
\label{subsec:repair2}

Let $(u^{\ast}, v^{\ast}) \in [r] \times [s]$ be the tuple associated with the code block to be repaired. Note that we need to reconstruct the $r$ symbols
$\big\{c(1;(u^{\ast}, v^{\ast})), c(2;(u^{\ast}, v^{\ast})),\ldots, c(r; (u^{\ast}, v^{\ast}))\big\}.$
We divide the repair process in the following two stages.
\begin{enumerate}
\item First, we recover the symbol $c(u^{\ast}; (u^{\ast}, v^{\ast}))$ using the {\rm Type~I} constraint containing it (cf.~\eqref{eq:type1}), i.e., 
\begin{align}
\label{eq:type1a2}
\sum_{(u,v) \in [r]\times [s]}c(u^{\ast}; (u,v)) = 0.
\end{align}
We download the $n-1$ symbols $$\big\{c(u^{\ast}; (u, v))~:~(u, v) \in [r] \times [s]~\text{s.t.}~(u, v) \neq (u^{\ast}, v^{\ast})\big\}$$ from the remaining $n - 1$ code blocks in this stage.
\item Next, we sequentially recover the $r-1$ symbols 
\begin{align}
\big\{c(u;(u^{\ast}, v^{\ast}))\big\}_{u \in [r]~\text{s.t.}~u \neq u^{\ast}}
\end{align} 
using the following $r-1$ {\rm Type~II} constraints (cf.~\eqref{eq:type2}).
\begin{align}
\label{eq:type2a2}
\underbrace{\sum_{({u}, {v}) \in [r]\times [s]}\lambda^p_{({u}-1)s + {v}} \cdot c(u^{\ast}; ({u},{v}))}_{\text{(a)}} + \underbrace{\sum_{{v} \in [s]}\psi \cdot c(\overline{u^{\ast}+p}; (u^{\ast}, {v}))}_{\text{(b)}} = 0,
\end{align}
where $p \in \{1,\ldots, r-1\}$. Note that the choice of {\rm Type~I} constraint used in the previous stage ensures that we now know all the values of the linear combinations in part (a) of these {\rm Type~II} linear constraints. Now assuming that $p \in \{1,\ldots, r-1\}$ is such that $\overline{u^{\ast} + p} = \hat{u} \in [r] \backslash \{u^{\ast}\}$, by downloading the additional $s - 1 = \frac{n}{r} - 1$ symbols
$
\big\{c(\hat{u};(u^{\ast}, {v})~:~{v} \in [s]~\text{s.t.}~{v} \neq v^{\ast}\big\}
$
which appear in the part $(b)$ of the linear constraint associated with the underlying value of $p$, we can recover the desired symbol $c(\hat{u};(u^{\ast}, v^{\ast}))$. Thus, the entire second stage involves downloading the following number of symbols (in addition to the symbols downloaded in the first stage).
\begin{align}
(r-1)(s-1) &= (r-1)\left({n}/{r} - 1\right) \leq r \left({n}/{r} - 1\right) = n - r \leq n - 1. \nonumber
\end{align}
\end{enumerate}
Note that the entire repair-by-transfer scheme described above downloads at most 
$
2(n-1) = 2\left(\frac{n-1}{n-k}\right)\cdot \ell
$
symbols of $\B$, which is twice the cut-set bound (cf.~\eqref{eq:cut-set}). 

\subsection{MDS property of the proposed codes.}
\label{subsec:mdsness2}

Next, we argue that Construction~\ref{subsec:construction2} gives us MDS linear array codes. This is equivalent to showing that for every $\Sc \subseteq [n]$ such that $|\Sc| = r = n-k$, the $r\ell \times r\ell$ sub-matrix $\Pm_{\Sc}$ of the parity-check matrix $\Pm$ (cf.~\eqref{eq:ParitySub}) is full rank. 

\begin{proposition}
\label{prop:perturbedMDS}
For given integers $n$ and $k$ such that $n > k$, the code obtained by Construction~\ref{subsec:construction2} is an MDS code.
\end{proposition}
\begin{proof}
See Appendix~\ref{appen:perturbedMDS}.
\end{proof}

It follows from Proposition~\ref{prop:perturbedMDS} that Construction~\ref{subsec:construction2} is a fully explicit construction of the exact-repairable MDS codes. However, this construction requires the size of the field $\B$ to be quite large. In particular, we need to have $|\B| \gg n^{(r - 1)\ell + 1}$. In contrast, by using a random selection of element $\psi$ in Construction~\ref{subsec:construction2} and employing standard techniques~\cite{SAK15}, one can show that a field of size $O(n^rr\ell)$ suffices. We note that even though this alternative approach requires a slightly smaller field, it does not give us a fully explicit construction.

%%%%%%%%%%%%%%%%%%%%%%%%%%%%%%%%%%%%%%%%%%%%%%%%%%%%%%%%%
%%%%%%%%%%%%%%%%%%%%%%%%%%%%%%%%%%%%%%%%%%%%%%%%%%%%%%%
% Construction for general t
%%%%%%%%%%%%%%%%%%%%%%%%%%%%%%%%%%%%%%%%%%%%%%%%%%%%%%%
%%%%%%%%%%%%%%%%%%%%%%%%%%%%%%%%%%%%%%%%%%%%%%%%%%%%%%%%%

\section{Construction of $(1 + \frac{1}{\tau}, (n-k)^{\tau}, n - 1)$-exact-repairable MDS Codes.}
\label{sec:RB_t}

In this section, we generalize the construction presented in Section~\ref{sec:twiceRB}. A design parameter $\tau$ allows us to increase the sub-packetization level $\ell$ in order to decrease the repair bandwidth of the code. In particular, for the integer $1 \leq \tau \leq \ceilb{{n}/{r}} - 1 = \ceilb{{n}/{(n-k)}} - 1$, we design exact-repairable MDS codes with sub-packetization level $\ell = r^{\tau} = (n-k)^{\tau}$, $t = n - 1$ and repair bandwidth at most 
$$
\left(1 + \frac{1}{\tau}\right)\left(\frac{n-1}{n - k}\right)\cdot \ell~~\text{symbols (over $\B$)}.
$$ 
This repair bandwidth is at most $\big(1 + \frac{1}{\tau}\big)$ times the cut-set bound (cf.~\eqref{eq:cut-set}). 

\begin{construction}
\label{subsec:construction_t}
Similar to Construction~\ref{subsec:construction2}, for ease of exposition, we assume that $r | n$ and $n = sr = s(n-k)$. We partition the $n$ code blocks in $r = n-k$ groups of equal sizes with each group containing $s = \frac{n}{r} = \frac{n}{n-k}$ code blocks. Using this partition, we index each code block by a tuple $(u, v)$ where $u \in [r] = [n-k]$ and $v \in [s]$. In particular, for $i \in [n]$ the associated tuple $(u, v)$ satisfies
$
i = (u-1)s + v.
$
Furthermore, we index the $\ell = r^{\tau} = (n-k)^{\tau}$ symbols (over $\B$) in each code block by the $r^{\tau}$ distinct $\tau$-length vectors in $[r]^{\tau} = [n-k]^{\tau}$. For $(u, v) \in [r] \times [s]$, the $\big((u-1)s + v\big)$-th code block can be represented as follows.
\begin{align*}
\cv_{(u-1)s + v} &= \cv_{(u, v)} = \big\{c\big((x_1, \ldots, x_{\tau}); (u, v)\big) \big\}_{(x_1, \ldots, x_{\tau}) \in [r]^{\tau}}.
\end{align*}
Similarly to Construction~\ref{subsec:construction2}, let $\Lf$ be finite field of size at least $n + 1$ and $ \{\lambda_i\}_{i \in [n]}$ be $n$ distinct non-zero elements of $\Lf$. Furthermore, $\B$ is an extension field of $\Lf$ such that the following two conditions hold.
\begin{itemize}
\item $\B$ is a simple extension of $\Lf$ which is generated by an element $\psi \in \B$, i.e., $\B = \Lf(\psi)$.
\item The degree of extension $[\B : \Lf]$ is at least $(r  - 1)\ell + 1$. 
\end{itemize}
We are now ready to present our construction of an $\big(1 + \frac{1}{\tau}, \ell = r^{\tau}, d = n - 1\big)$-exact-repairable MDS code $\Cc$ by defining an $r\ell \times n\ell$ parity-check matrix $\Pm$ of the code $\Cc$. Specifically, we classify the $r\ell = r^{\tau+1}$ linear constraints defined by the parity-check matrix $\Pm$ into two types. 
\begin{itemize}
\item {\bf {\rm Type~I} constraints:}~We have $\ell = r^{\tau}$ {\rm Type~I} constraints which are defined by the first $\ell = r^{\tau}$ rows of the matrix $\Pm$. For every $(x_1, \ldots, x_{\tau}) \in [r]^{\tau}$, we have
\begin{align}
\label{eq:type1t}
\sum_{(u,v) \in [r]\times [s]}c\big((x_1,\ldots, x_{\tau}); (u, v)\big) = 0.
\end{align}
\item {\bf {\rm Type~II} constraints:}~We have $(r-1)\ell = (r-1)r^{\tau}$ {\rm Type~II} constraints. Recall that for strictly positive integers $e$ and $m$, the quantity $\overline{e}^{\{m\}}$ is defined as follows.
\begin{align}
\overline{e}^{\{m\}} = \begin{cases}
m &\mbox{if}~e~({\rm mod}~m) = 0, \\
e~({\rm mod}~m) &\mbox{otherwise}.
\end{cases}
\end{align} 
Assuming that $v \in [s] = \left[\frac{n}{r}\right]$ be such that $\overline{v}^{\{\tau\}} = a \in [\tau]$ and $p \in \{1, 2,\ldots, r-1\}$, we use $\overline{\xv}^{v, p} = \overline{(x_1,\ldots, x_{\tau})}^{v, p}$ to denote the vector obtained by modifying a single coordinate of the vector $\xv = (x_1,\ldots, x_{\tau})$ in the following manner. 
\begin{align}
&\overline{\xv}^{v, p} = \overline{(x_1,\ldots, x_{\tau})}^{v, p} \nonumber \\
&~~~~= (x_1,\ldots, x_{a-1}, \overline{x_{\overline{v}^{\{\tau\}}} + p}^{\{r\}},x_{a+1},\ldots, x_{\tau}) \nonumber \\
&~~~~= (x_1,\ldots, x_{a-1}, \overline{x_{a} + p}^{\{r\}},x_{a+1},\ldots, x_{\tau}). \nonumber 
\end{align}
For every $p \in \{1,\ldots, r - 1\}$ and $(x_1, \ldots, x_{\tau}) \in [r]^{\tau}$, we have an associated linear constraint in the parity-check matrix $\Pm$.
\begin{align}
\label{eq:type2t}
&\underbrace{\sum_{(u, v) \in [r]\times [s]}\lambda^p_{(u-1)s + v} \cdot c\big((x_1, \ldots, x_{\tau}); (u, v)\big)}_{\text{(a)}} + \underbrace{\sum_{v \in [s]}\psi \cdot c\big(\overline{(x_1, \ldots, x_{\tau})}^{v, p}; (x_{\overline{v}^{\{\tau\}}}, v)\big)}_{\text{(b)}} = 0. 
\end{align}
We can partition the {\rm Type~II} constraints (cf.~\eqref{eq:type2t}) into $(r - 1)$ groups (each group containing $\ell = r^{\tau}$ linear constraints) according to the value of $p \in \{1,\ldots, r-1\}$. In particular, $r^{\tau}$ constraints associated with the same value of $p$ constitute those $r^{\tau}$ rows of the parity-check matrix $\Pm$ that are indexed by the set $$\{pr^{\tau} + 1,\ldots, (p+1)r^{\tau}\} \subseteq [r\ell] = [r^{\tau+1}].$$
\end{itemize}
\end{construction}

\begin{example}
\label{ex:type2t}
In this example, we look at the composition of a {\rm Type~II} constraint (cf.~\ref{eq:type2t}) when $\tau=2$. We assume that $n = 9$ and $r = n - k = 3$. This implies that $s = \frac{n}{n-k} = 3$. For $(x_1, x_2) \in [r]^{\tau} =  [3]^2$ and $p = 1$ the associated {\rm Type~II} constraint takes the following form.
\begin{align}
&\sum_{(u, v) \in [3]\times [3]}\lambda_{(u-1)3 + v} \cdot c\big((x_1, x_2); (u, v)\big)~~+ \nonumber \\ 
&~~~~~~~~~~~~\psi \cdot \Big( c\big((\overline{x_1 + 1}^{\{3\}}, x_2); (x_1, 1)\big) + c\big((x_1, \overline{x_2 + 1}^{\{3\}}); (x_2, 2)\big) + c\big((\overline{x_1 + 1}^{\{3\}}, x_2); (x_1, 3)\big) \Big) = 0. \label{eq:type2t_ex}
\end{align}
Note that we have used the following equalities in \eqref{eq:type2t_ex} which hold for $\tau = 2$ and $s = \frac{n}{n-k} = 3$.
\begin{align*}
\overline{1}^{\{\tau = 2\}} = \overline{3}^{\{2\}} = 1~~\text{and}~~\overline{2}^{\{2\}} = 2.
\end{align*}
\end{example}

\subsection{Exact repair of failed code blocks in the proposed codes.}
\label{subsec:repair_t}

We now illustrate a mechanism to perform exact repair of a code block in the code obtained by Construction~\ref{subsec:construction_t}. Before we describe the repair mechanism, let's introduce some notation. For a vector $\xv = (x_1,\ldots, x_{\tau}) \in [r]^{\tau}$ and $a \in [\tau]$, we use 
$$\xv_{\backslash\{a\}} = (x_1,\ldots, x_{a-1}, x_{a+1},\ldots, x_{\tau}) \in [r]^{\tau-1}$$ 
to denote the vector obtained by puncturing the $a$-th coordinate of the vector $\xv$. Similarly, for a vector $\xv = (x_1, \ldots, x_{\tau}) \in [r]^{\tau}$, $a \in [\tau]$ and $u \in [r]$, we use
$$
\xv_{\{a, u\}} = (x_1,\ldots,x_{a-1}, x_{a} = u, x_{a+1},\ldots, x_{\tau}) \in [r]^{\tau}
$$
to denote the vector obtained by replacing the $a$-th coordinate of the vector $\xv$ by $u$. 

Let $(u^{\ast}, v^{\ast}) \in [r] \times [s]$ be the tuple associated with the code block to be repaired. Note that we need to reconstruct the following $\ell = r^{\tau}$ code symbols.
\begin{align}
\big\{c\big((x_1,\ldots, x_{\tau});(u^{\ast}, v^{\ast})\big)\big\}_{(x_1, \ldots, x_{\tau}) \in [r]^{\tau}}.
\end{align}
Similar to Section~\ref{subsec:repair2}, we divide the repair process in the following two stages.
\begin{enumerate}
\item In the first stage we utilize {\rm Type~I} constraints (cf.~\eqref{eq:type1t}) to recover the following $r^{\tau-1}$ symbols.
\begin{align}
\label{eq:stage1t}
\big\{c\big(\xv_{\{a, u^{\ast}\}};(u^{\ast}, v^{\ast})\big)\big\}_{\xv_{\backslash\{a\}} \in [r]^{\tau-1}},
\end{align} 
where $a = \overline{v^{\ast}}^{\{\tau\}}$. Recall that for $\xv = (x_1,\ldots, x_{t}) \in [r]^{\tau}$, the {\rm Type~I} constraint takes the following form.
\begin{align}
\label{eq:type1a}
\sum_{(u,v) \in [r]\times [s]}c\big((x_1, \ldots, x_{\tau})); (u,v)\big) = 0.
\end{align}
Therefore, in order to recover the $r^{\tau-1}$ symbols shown in \eqref{eq:stage1t} using these constraints, we download the following $(n-1)r^{\tau-1}$ symbols from the remaining $n-1$ code blocks. 
\begin{align}
\label{eq:repair_stage1}
&\big\{c\left(\xv_{\{a, u^{\ast}\}}; ({u}, {v})\right)\big\}_{\xv_{\backslash\{a\}} \in [r]^{\tau-1}, ({u}, {v}) \neq (u^{\ast}, v^{\ast})},
\end{align}
where $a = \overline{v^{\ast}}^{\{\tau\}}$. Recall that the tuples $(u, v)$ and $(u^{\ast}, v^{\ast})$ take values in the set $[r]\times [s]$.
\item At the end of the stage $1$ of the repair process, we have access to the following symbols which also include the $r^{\tau-1}$ symbols recovered in the stage $1$. 
\begin{align}
\label{eq:repair_stage1}
&\big\{c\left(\xv_{\{a, u^{\ast}\}};(u, v)\right)\big\}_{\xv_{\backslash\{a\}} \in [r]^{\tau-1}, ({u}, {v}) \in [r] \times [s]}.
\end{align}
In the stage $2$ of the repair process, we employ the {\rm Type~II} constraints to sequentially recover the remaining $(r-1)r^{\tau-1}$ symbols
\begin{align}
\label{eq:stage2t}
\big\{c\big(\xv_{\{a, u\}};(u^{\ast}, v^{\ast})\big)\big\}_{\xv_{\backslash\{a\}} \in [r]^{\tau-1}, u \neq u^{\ast}},
\end{align}
where $a = \overline{v^{\ast}}^{\{\tau\}}$. Recall that both $u$ and $u^{\ast}$ take values in the set $[r]$. Let $p \in \{1,2,\ldots, r-1\}$ be such that we have $\overline{u^{\ast} + p}^{\{r\}} = \hat{u} \in [r]\backslash\{u^{\ast}\}$. We utilize the following {\rm Type~II} constraint to repair the desired symbol $c\big(\xv_{\{a, \hat{u}\}};(u^{\ast}, v^{\ast})\big) = c\big((x_1,\ldots, x_{a-1}, \hat{u}, x_{a+1},\ldots, x_{\tau});(u^{\ast}, v^{\ast})\big).$
\begin{align}
&~~~~~~~~~~\underbrace{\sum_{({u}, {v}) \in [r]\times [s]}\lambda^p_{(u-1)s + v} \cdot c\big(\xv_{\{a, u^{\ast}\}}; ({u}, {v})\big)}_{\text{(a)}} + \nonumber \\
&\underbrace{\psi \cdot c\big(\xv_{\{a, \hat{u}\}}; (x_{a} = u^{\ast}, v^{\ast})\big)}_{\text{(b-I)}} + \underbrace{\sum_{v \in [s]~:~v \neq v^{\ast}}\psi \cdot c\big(\overline{\xv_{\{a, u^{\ast}\}}}^{v, p}; (x_{\overline{v}^{\{\tau\}}}, v)\big)}_{\text{(b-II)}} = 0. \label{eq:type2at}
\end{align}
It is straightforward to verify that at the end of the stage $1$ of the repair process, we know the value of the linear combination in the part (a) of this linear constraint (cf.~\eqref{eq:repair_stage1}). We now argue that we also know many of the symbols appearing in the part (b-II) of this constraint. Note that the part (b-II) can be rewritten as follows.
\begin{align}
&\sum_{v \in [s]~:~v \neq v^{\ast}}\psi \cdot c\big(\overline{\xv_{\{a, u^{\ast}\}}}^{v, p}; (x_{\overline{v}^{\{\tau\}}}, v)\big) \nonumber\\
&~~~~~~~~~= \underbrace{\sum_{v \neq v^{\ast}~:~\overline{v}^{\{\tau\}} =\overline{v^{\ast}}^{\{\tau\}} = a}\psi \cdot c\big(\xv_{\{a, \hat{u}\}}; (x_{a} = u^{\ast}, v)\big)}_{\text{(b-II-1)}} + \underbrace{\sum_{v \neq v^{\ast}~:~\overline{v}^{\{\tau\}} \neq \overline{v^{\ast}}^{\{\tau\}} = a}\psi \cdot c\big(\overline{\xv_{\{a, u^{\ast}\}}}^{v, p}; (x_{\overline{v}^{\{\tau\}}}, v)\big)}_{\text{(b-II-2)}}. \nonumber
\end{align}
Note that the code symbols appearing in part (b-II-2) are indexed by the vectors which have their $a$-th coordinate equal to $u^{\ast}$. One can verify that these symbols are already known at the end of the stage $1$ of the repair process (cf.~\eqref{eq:repair_stage1}). Therefore, in order to recover the desired symbol $$c\big((x_1,\ldots, x_{a-1}, \hat{u}, x_{a+1},\ldots, x_{\tau});(u^{\ast}, v^{\ast})\big)$$ using the linear constraint in \eqref{eq:type2at}, we need to only download the code symbols appearing in the part (b-II-1). Note that there are at most $\floorb{\frac{s}{\tau}}$ symbols in the part (b-II-1). Since we have to repair $(r-1)r^{\tau-1}$ symbols in the stage $2$ (cf.~\eqref{eq:stage2t}), the number of symbols that we download in the stage $2$ (in addition to the symbol downloaded in the stage $1$) is at most
\begin{align}
&(r-1)r^{\tau-1}\floorb{\frac{s}{\tau}} \leq (r-1)r^{\tau-1}\left({\frac{s}{\tau}}\right) = \frac{r^{\tau-1}}{\tau}\frac{r-1}{r}n \overset{(i)}{\leq}\frac{r^{\tau-1}}{\tau}(n-1). \nonumber
\end{align}
Here the step $(i)$ follows as, for $r = n- k \leq n$, we have $\frac{r-1}{r} \leq \frac{n-1}{n}$. Since we download $(n-1)r^{\tau-1}$ symbols during the stage $1$ of the repair process, the total repair bandwidth is at most 
\begin{align}
&(n-1)r^{\tau-1} + \frac{r^{\tau-1}}{\tau}(n-1) = \left(1+ \frac{1}{\tau}\right)(n-1)r^{\tau-1} =  \left(1+ \frac{1}{\tau}\right)\left(\frac{n-1}{n-k}\right)\cdot \ell, \nonumber
\end{align}
which is $\left(1 + 1/\tau\right)$ times the cut-set bound (cf.~\eqref{eq:cut-set}).
\end{enumerate}

\begin{remark}
\label{rem:lowerRB_t}
Note that the advantage of having higher sub-packetization is realized in the second stage of the repair process. It allows us to design part (b) of {\rm Type~II} constraints (cf.~\eqref{eq:type2a2}) in a manner so that most of symbols appearing in the part (b) of the {\rm Type~II} constraints used in the second repair stage are already known at the end of the first stage. In particular, we only need to download at most ${\frac{1}{\tau}}$-th fraction of the symbols appearing in the part (b) of a {\rm Type~II} constraint. This implies that, as we work with the higher values of $\tau$, we get further reduction in the repair bandwidth of the obtained codes.
\end{remark}

\begin{remark}[MDS property of the proposed codes]
\label{subsec:mds_t}
The argument for this part is identical to that used in Section~\ref{subsec:mdsness2}.
\end{remark}

\section{Repair-efficient Linear Array Codes with Small Sub-packetization Levels.}
\label{sec:gen_approach}

In this section, we present a general approach to realize our end goal of constructing $\epsilon$-MSR codes with small sub-packetization levels. 
\begin{comment}
{\color{red}Note that for $\epsilon = 0$, the $\epsilon$-MSR code corresponds to an MDS code that achieves optimal repair bandwidth. However, an MSR code (or $\epsilon$-MSR code with $\epsilon = 0$) requires a large sub-packetization level~\cite{GTC14,TWB14}.} Here, we focus on a special family of MSR codes (with large sub-packetization level) and combine them with codes with large enough distance in order to obtain {\em long} $\epsilon$-MSR codes with small sub-packetization levels. 
\end{comment}
Towards this, we combine a short MSR code with another code that has large minimum distance to obtain a linear array code which has a significantly small sub-packetization level as a function of its length. In particular, for a constant number of parity blocks, it is possible to obtain codes of length which is exponential in their sub-packetization levels. Furthermore, the repair bandwidth of the obtained code is only slightly larger than that of an MSR code with the same parameters. In Section~\ref{sec:long_mds}, we utilize a family of MSR codes from \cite{YeB16a} to ensure that the obtained long code is an MDS code. As a result, the approach described in this section gives us $\epsilon$-MSR codes with small sub-packetization levels.

\begin{construction}
\label{cons:generic}
We are given two codes $\Cc^{\rm I}$ and $\Cc^{\rm II}$.
\begin{enumerate}
\item $\Cc^{\rm I}$ is an  $(n = k + r,  k, t = n - 1, \ell)_{\B}$ MSR code defined by the parity-check matrix
\begin{align}
\label{eq:parityI}
\Hm = \left(
\begin{array}{cccc}
H_{1,1} & H_{1, 2} & \cdots & H_{1, n} \\
\vdots & \vdots & \ddots & \vdots \\
H_{r, 1} & H_{r, 2} & \cdots & H_{r, n}
\end{array}
\right) \in \B^{r\ell \times n\ell}.
\end{align}
For $i \in [n]$, the repair matrix associated with the $i$-th code block takes the following diagonal form. 
\begin{align}
\label{eq:repairI}
S_i = \Diag\left(S_{i, 1},S_{i, 2},\ldots,S_{i, r} \right) \in \B^{\ell \times r\ell},
\end{align}
where for each $j \in [r]$, $S_{i, j}$ is an $\frac{\ell}{r} \times \ell$ matrix (over $\B$).
\item $\Cc^{\rm II}$ is an $(N, M, D = \delta N)_{\G}$ code defined over an alphabet $\G$ of size at most $n$.
\end{enumerate}
Given these two codes, we construct an $[\Nc = M, \Kc l = (M - r)l, \Dc, l = N\ell]_{\B}$ linear array code $\Cc = \Cc^{\rm II}\circ\Cc^{\rm I}$ by designing its $rN\ell \times MN\ell$ parity-check matrix $\Hc$.
Note that a codeword of $\Cc$ comprises $M = |\Cc^{\rm II}|$ code blocks with each of these blocks containing $N\ell$ symbols of $\B$. The $M$ code blocks in a codeword of $\Cc$ are indexed by $M$ distinct $N$-length codewords in $\Cc^{\rm II}$. Let $\cv = (c_1, \ldots, c_N) \in \G^N$ be a codeword of $\Cc^{\rm II}$. Then, the $N\ell$ columns of the parity-check matrix $\Hc$ that correspond to the code block of a codeword of $\Cc$ indexed by $\cv \in \Cc^{\rm II}$ are defined as follows. 
\begin{align}
\label{eq:thickCol}
\mathbf{\Hc}_{\cv} 
= \left[
\begin{array}{c}
{\alpha_{1,\cv}} \cdot \Diag(H_{1,c_1}, \ldots, H_{1, c_N}) \\
\vdots \\
{\alpha_{r,\cv}} \cdot \Diag(H_{r,c_1}, \ldots, H_{r, c_N}) 
\end{array}\right],
\end{align}
where $\{\alpha_{j, \cv}\big\}_{j \in [r], \cv \in \Cc^{\rm II}}$ are non-zero elements from $\B$. We associate the alphabet $\G$ with the set of integers $\{1, 2,\ldots,|\G|\}$ while specifying the parity-check matrix $\Hc$ in \eqref{eq:thickCol}. Note that all the blocks $\{H_{j, c_i}\}_{j \in [r], i \in [N]}$ in \eqref{eq:thickCol} are well defined as we have $|\G| \leq n$.   
\end{construction}

{As shown in Section~\ref{sec:long_mds}, depending on the specific choice of the MSR code $\Cc^{\rm I}$, these scalars can be chosen to ensure that the obtained code $\Cc$ is an MDS code.} Next, we show that the code $\Cc$ obtained from Construction~\ref{cons:generic} has a linear repair scheme with small repair bandwidth, regardless of the choice for these scalars.

\begin{theorem}
\label{thm:main}
The code $\Cc$ in Construction~\ref{cons:generic} is an $[\Nc = M, (\Nc - r)N\ell, \Dc, N\ell]_{\B}$ linear array code which enables repair of every code block in each of its codewords by downloading at most
$$
\big(1 + (r-1)(1 - \delta)\big)\cdot{N\ell}/{r}
$$
symbols of $\B$ from each of the remaining $\Tc = \Nc - 1$ code blocks.
\end{theorem}

Before presenting a proof of Theorem~\ref{thm:main}, we make the following observations regarding Construction~\ref{cons:generic}.

\begin{remark}
If the code $\Cc$ in Construction~\ref{cons:generic} is an MDS code, then it follows from Theorem~\ref{thm:main} that $\Cc$ is an $\epsilon$-MSR code (cf.~Definition~\ref{def:eMSR}) with $\epsilon=(r-1)(1-\delta)/r$, where $\delta$ is the relative minimum distance of $\Cc^{\rm II}$.
\end{remark}

\begin{remark}
Note that the larger $\delta$ is, the smaller $\epsilon$ we get in Construction~\ref{cons:generic}. However, taking large $\delta$ reduces the size of the code $\Cc^{\rm II}$, which further implies that we get a shorter code $\Cc$. Namely, there is a clear trade-off between $\epsilon$ and the length of the code $\Cc$.
\end{remark}

%\begin{proof}
\noindent {\em Proof of Theorem~\ref{thm:main}.~}For $i \in [M]$, we demonstrate a linear repair scheme for the $i$-th code block of a codeword of $\Cc$. Recall that the $M$ code blocks in a codeword of $\Cc$ are indexed by $M$ distinct codewords in the code $\Cc^{\rm II}$. Let the code block to be repaired be indexed by the codeword 
$
\cv = (c_1, c_2,\ldots, c_N)  \in \Cc^{\rm II}.
$
We claim that the following $N\ell \times rN\ell$ matrix serves as a repair matrix for this code block.
\begin{align}
\Sc_{\cv} &= \Diag\big(\Diag(S_{c_1, 1}, \ldots, S_{c_N, 1}),\cdots,\Diag(S_{c_1, r}, \ldots, S_{c_N, r}) \big) \nonumber
\end{align} 
It is sufficient to verify the conditions given in \eqref{eq:Repair_cond1}, i.e., 
\begin{align}
\rank\left(\Sc_{\cv} \left[\begin{array}{c}
\Hc_{1, \cv} \\
\vdots \\
\Hc_{r, \cv}
\end{array}
\right] \right) = N\ell.
\end{align}
Recall that, as per our assumption, $\cv$ denotes the $i$-th codeword of the code $\Cc^{\rm II}$. Note that we have
\begin{align}
\Sc_c \left[\begin{array}{c}
\Hc_{1, \cv} \\
\vdots \\
\Hc_{r, \cv}
\end{array}
\right]  
&=  \Diag\big(\Diag(S_{c_1, 1}, \ldots, S_{c_N, 1}),\cdots, \Diag(S_{c_1, r}, \ldots, S_{c_N, r}) \big)
\left(
\begin{array}{c}
\alpha_{1,\cv} \cdot \Diag(H_{1,c_1},\ldots, H_{1, c_N}) \\
\vdots \\
\alpha_{r,\cv} \cdot \Diag(H_{r,c_1}, \ldots, H_{r, c_N})
\end{array}\right) \nonumber \\
& = \left(
\begin{array}{c}
\alpha_{1,\cv} \cdot \Diag(S_{c_1, 1}H_{1,c_1},\ldots, S_{c_N, 1}H_{1, c_N}) \\
\vdots \\
\alpha_{r,\cv} \cdot \Diag(S_{c_1, r}H_{r,c_1}, \ldots, S_{c_N, r}H_{r, c_N})
\end{array}\right)
\end{align}
Therefore, we have 
\begin{align}
\rank\left(\Sc_{\cv} \left[\begin{array}{c}
\Hc_{1, \cv} \\
\vdots \\
\Hc_{r, \cv}
\end{array}
\right] \right) 
& = \sum_{j = 1}^{N} \rank \left(\left[\begin{array}{c}
S_{c_j, 1}H_{1, c_j} \\
\vdots \\
S_{c_j, r}H_{r, c_j}
\end{array}
\right]\right) \nonumber \\
& \overset{(i)}{=} \sum_{j = 1}^{N} \rank \left(S_{c_j}\Hm_{c_j}\right) \overset{(ii)}{=} \sum_{j = 1}^{N}\ell = N\ell,
\end{align}
where $(i)$ follows from the structure of the repair matrix $S_i$ in the short MSR code $\Cc^{\rm I}$ (cf.~\eqref{eq:repairI}) and $(ii)$ follows from the requirement on the repair matrices of $\Cc^{\rm I}$ (cf.~\eqref{eq:Repair_cond1}).

\noindent \textbf{Repair bandwidth:} Next, we focus on the repair bandwidth associated with the repair matrix $\Sc_{\cv}$. For a codeword $\tilde{\cv} = (\tilde{c}_1, \tilde{c}_2,\ldots, \tilde{c}_N) \in \Cc^{\rm II}$ such that $\tilde{\cv} \neq \cv$, the code block in a codeword of $\Cc$ which is indexed by $\widetilde{\cv}$ needs to contribute
\begin{align}
\rank\left(\Sc_{\cv} \left[\begin{array}{c}
\Hc_{1, \tilde{\cv}} \\
\vdots \\
\Hc_{r, \tilde{\cv}}
\end{array}
\right]  \right) \nonumber 
\end{align} 
symbols during the repair of the code block indexed by $\cv \in \Cc^{\rm II}$. Note that
\begin{align}
\label{eq:repairBW_long_S}
\Sc_{\cv} \left[\begin{array}{c}
\Hc_{1, \tilde{\cv}} \\
\vdots \\
\Hc_{r, \tilde{\cv}}
\end{array}
\right]  
&=  \Diag\big( \Diag(S_{c_1, 1}, \ldots, S_{c_N, 1}),\cdots, \Diag(S_{c_1, r}, \ldots, S_{c_N, r}) \big)
\left(
\begin{array}{c}
\alpha_{1,\tilde{\cv}} \cdot \Diag(H_{1,\tilde{c}_1},\ldots, H_{1, \tilde{c}_N}) \\
\vdots \\
\alpha_{r,\tilde{\cv}} \cdot \Diag(H_{r, \tilde{c}_1}, \ldots, H_{r, \tilde{c}_N})
\end{array}\right) \nonumber \\ 
&=\left(
\begin{array}{c}
\alpha_{1,\tilde{\cv}} \cdot \Diag(S_{c_1, 1}H_{1,\tilde{c}_1},\ldots, S_{c_N, 1}H_{1, \tilde{c}_N}) \\
\vdots \\
\alpha_{r,\tilde{\cv}} \cdot \Diag(S_{c_1, r}H_{r,\tilde{c}_1}, \ldots, S_{c_N, r}H_{r, \tilde{c}_N})
\end{array}\right)
\end{align}
Therefore, we have
\begin{align}
\label{eq:repairBW_long}
\rank\left(\Sc_{\cv} \left[\begin{array}{c}
\Hc_{1, \tilde{\cv}} \\
\vdots \\
\Hc_{r, \tilde{\cv}}
\end{array}
\right]  \right) 
& = \sum_{j = 1}^{N} \rank \left(\left[\begin{array}{c}
S_{c_j, 1}H_{1, \tilde{c}_j} \\
\vdots \\
S_{c_j, r}H_{r, \tilde{c}_j}
\end{array}
\right]\right) \overset{(i)}{=} \sum_{j = 1}^N\rank(S_{c_j}\Hm_{\tilde{c}_j}),
\end{align}
where $(i)$ follows from \eqref{eq:repairI}. We now consider two cases.
\begin{enumerate}
\item {\bf Case~$1$~($\widetilde{c}_j = c_j$):}~In this case, we have
\begin{align}
\label{eq:beta_case1}
\rank(S_{c_j}\Hm_{\tilde{c}_j}) = \rank(S_{c_j}\Hm_{{c}_j})  \overset{(i)}{=} \ell,
\end{align}
where $(i)$ follows from \eqref{eq:Repair_cond1}.
\item {\bf Case~$2$~($\widetilde{c}_j \neq c_j$):}~Note that we have
\begin{align}
\label{eq:beta_case2}
\rank(S_{c_j}\Hm_{\tilde{c}_j}) \overset{(i)}{=} \frac{\ell}{r},
\end{align}
where $(i)$ follows from the fact that $S_{c_i}$ is the repair matrix for the $c_i$-th code block of the MSR code $\Cc^{\rm I}$. Recall that we have $c_i \in \G$ is associated with an element of $\{1, 2,\ldots,|\G|\} \subseteq [n]$.
\end{enumerate}
By substituting \eqref{eq:beta_case1} and \eqref{eq:beta_case2} in \eqref{eq:repairBW_long}, we obtain that
\begin{align}
\rank\left(\Sc_{\cv} \left[\begin{array}{c}
\Hc_{1, \tilde{\cv}} \\
\vdots \\
\Hc_{r, \tilde{\cv}}
\end{array}
\right]  \right) & = \sum_{j \in [N] : c_j = \tilde{c}_j}\rank(S_{c_j}\Hm_{\tilde{c}_j})  +  \sum_{j \in [N] : c_j \neq \tilde{c}_j} \rank(S_{c_j}\Hm_{\tilde{c}_j}) \nonumber \\
& =  |\{j \in [N] : c_j = \tilde{c}_j\}|\ell + |\{j \in [N] : c_j \neq \tilde{c}_j\}|\frac{\ell}{r} \nonumber \\
& =  N\ell - |\{j \in [N] : c_j \neq \tilde{c}_j\}|\left(\frac{r-1}{r}\right)\ell \label{eq:epsilon_avg} \\
& \leq N\ell  - D\left(\frac{r-1}{r}\right)\ell \nonumber \\
&= \frac{N\ell}{r} + \left(\frac{r - 1}{r}\right)(N - D)\ell \nonumber \\
&{=}\big(1 + (r-1)(1 - \delta)\big)\cdot\frac{N\ell}{r}, \label{eq:epsilon_bw_bound}
\end{align}
where \eqref{eq:epsilon_bw_bound} follows from $D=\delta N$. 
\qed

\begin{remark}
\label{rem:generic}
For a given $\epsilon > 0$, one can choose $\Cc^{\rm II}$ to be such that 
$$
\delta = \frac{D}{N} \geq 1 - \frac{\epsilon}{r-1}.
$$ Combining this with \eqref{eq:epsilon_bw_bound} gives us that 
\begin{align}
\label{eq:epsilon_bw_bound1}
\rank\left(\Sc_{\cv} \left[\begin{array}{c}
\Hc_{1, \tilde{\cv}} \\
\vdots \\
\Hc_{r, \tilde{\cv}}
\end{array}
\right]  \right) \leq \big(1 + \epsilon\big)\cdot\frac{N\ell}{r}.
\end{align} 
This implies that the repair bandwidth of the code $\Cc = \Cc^{\rm II}\circ\Cc^{\rm I}$ is at most $(1 + \epsilon)$ times the lower bound on the repair bandwidth of an MDS code with the same parameters (cf.~\eqref{eq:cut_set_d}).
\end{remark}

\begin{remark}
\label{rem:avg_bw}
Note that \eqref{eq:epsilon_bw_bound1} provides an upper bound on the number of symbols transmitted from each node. However assuming that $\Cc^{II}$ is a linear code, one can give a tighter bound on the average number symbols transmitted from each node. In particular, for $\Cc^{\rm II}$ which is a linear code over an alphabet of size $q$, its average distance satisfies the following.
$$
\overline{d} = \frac{1}{|\Cc^{\rm II}|(|\Cc^{\rm II}|-1)}\sum_{\cv, \cv' \in \Cc:~\cv \neq \cv'}\dist(\cv, \cv')  = \frac{(q-1)|\Cc^{\rm II}|}{q(|C^{\rm II}| - 1)}N \geq \left(1 - {1}/{q}\right)N.
$$
This can be combined with \eqref{eq:epsilon_avg} to conclude that the obtained linear array code downloads on average $\left(1 + (r-1)/q\right)\cdot\frac{N\ell}{r}$ symbols of $\B$ from an intact code block. Note that it follows from the Plotkin bound that $1/q \leq 1 - \delta$. Thus, the aforementioned bound on the average number of symbols download from a code block is better that the one given in \eqref{eq:epsilon_bw_bound1}.
\end{remark}

\begin{remark}
\label{rem:RBT_long}
It is straightforward to verify from \eqref{eq:repairBW_long_S} that if $\Cc^{\rm I}$ has a repair-by-transfer scheme (cf.~Remark~\ref{rem:RBT}), then $\Cc = \Cc^{\rm II}\circ\Cc^{\rm I}$ from Construction~\ref{cons:generic} also possesses a repair-by-transfer scheme.
\end{remark}

In Table~\ref{tab:epsilon_MSR}, we illustrate the repair bandwidth of linear array codes obtained by Construction~\ref{cons:generic} with the help of a few examples. We employ the MSR codes obtained in \cite{WTB16} as the short codes $\Cc^{\rm I}$ in these examples. Moreover, all the code used as $\Cc^{\rm II}$ in these examples are linear codes. This allows us to obtain upper bound on the average number of symbols downloaded from a code block stored on an intact node (cf.~Remark~\ref{rem:avg_bw}). In order to get an idea of the parameters realized by our approach, let's consider the second example in Table~\ref{tab:epsilon_MSR} where we obtain a code with length $\Nc = 81$ and $r = 2$ parity nodes. The obtained code achieves $\frac{\beta}{l} = 0.55$ which is slightly worse than the optimal value of $0.5$, achieved by an MSR code of the same length and code rate. However, the obtained code has its sub-packetization level equal to $80$ as compared to the (best known) MSR code that requires the sub-packetization level to be a prohibitively large value of $2^{27}$.
\bgroup
\def\arraystretch{1.55}
\begin{table*}[t!]
\footnotesize
\centering
\begin{tabular}{c| c| c| c| c |c| c| c| }
  \hline \hline
\pbox{3.5cm}{~~~~~~~~~~~~~~~~$\cC^{\rm I}$ \\ $(n,k,t=n-1,\ell)$ MSR code}& \pbox{2.8cm}{~~~~~~~~~$\cC^{\rm II}$ \\ $(N,K,D=\delta N)_q$ code} & $\Nc=q^K$ & $r$ &  $\Kc=q^K-r$ & ${\beta}/{l}$ & $\overline{\beta}/l$ & $l = N\ell = N\cdot r^\frac{q}{r+1}$  \\
	\hline
	$(3,1,2,2)$ MSR code &		$[20,3,13]_3$&   $27$   & $2$&   $25$ &  $0.675~(0.5)$ 			 & $0.653$  & $40~(2^9)$\\ \hline
	$(9,7,8,8)$ MSR code&		$[10,2,9]_9$ &   $81$   & $2$&   $79$ &  $0.55~(0.5)$ 				 & $0.55$  & $80~(2^{27})$\\ \hline
	$(9,7,8,8)$ MSR code&   $[15,3,12]_9$&   $729$  & $2$&   $727$&  $0.6~(0.5)$          & $0.554$ & $120~(2^{81})$\\ \hline
	$(8,5,7,9)$ MSR code&   $[20,3,16]_8$&   $512$  & $3$&   $509$&  $0.466~(0.33)$ & $0.415$  & $180~(3^{128})$\\ \hline
\end{tabular}
\caption{Examples of linear array codes obtained using Construction~\ref{cons:generic}. The short codes $\Cc^{\rm I}$ utilized in these examples are constructed by Wang et al.~\cite{WTB16}.  $\beta$ and $\overline{\beta}$ denote upper bounds on the maximum number of symbol and the average number of symbols downloaded from an intact code block during the repair process, respectively. The term inside the brackets in the third column from the right denotes the value of $\beta/l$ achieved by the MSR code with $r$ parity nodes. The term inside the brackets in the rightmost column represents the sub-packetization level needed by the best known constructions for the MSR codes with parameters $\Nc, \Kc$ and $\Tc = \Nc - 1$. }
\label{tab:epsilon_MSR}
\end{table*}
\egroup

\section{Instantiation of the $\epsilon$-MSR Code Construction.}
\label{sec:long_mds}

In this section, we further explore Construction~\ref{cons:generic} in terms of the code parameters that can be achieved by specific choices for two constituents $\Cc^{\rm I}$ and $\Cc^{\rm II}$ of the construction. First, we give an explicit construction of $\epsilon$-MSR code using a specific family of MSR codes by Ye and Barg~\cite{YeB16a} as $\Cc^{\rm I}$. As described in Section~\ref{sec:gen_approach}, for any $\Cc^{\rm II}$ with large minimum distance, the code $\Cc$ obtained by Construction~\ref{cons:generic} always has small repair bandwidth. Choosing Ye and Barg codes~\cite{YeB16a} as $\Cc^{\rm I}$ allows us to argue that $\Cc$ is an MDS code as well, which is the other requirement for a code to be an $\epsilon$-MSR code. We then focus on the specific choice for $\Cc^{\rm II}$. In particular, we utilize the existence of codes on the GV curve to establish the existence of $\epsilon$-MSR codes with good relationship between the length and the sub-packetization level.

\subsection{Explicit $\epsilon$-MSR codes using Ye and Barg construction}
\label{sec:composing}
In \cite{YeB16a}, Ye and Barg present an explicit construction of MSR codes with sub-packetization level $\ell = (n-k)^n$, which is exponential in the code length $n$. We now illustrate how one can select the MSR codes obtained by this construction as $\Cc^{\rm I}$ in Construction~\ref{cons:generic} in order to design explicit $\epsilon$-MSR codes with small sub-packetization level. For the necessary details of the construction of Ye and Barg~\cite{YeB16a}, we refer the reader to Appendix~\ref{appen:YeBarg} where we briefly describe the construction along with the associated repair scheme.
 
Let $\Cc^{\rm II}$ be the $(N, M = q^{NR} = |\Cc^{\rm II}|, D)$ code that we combine with the MSR code $\Cc^{\rm I}$ in Construction~\ref{cons:generic}. Let $\B$ be a finite field of size greater than $|C^{\rm II}|rn$, i.e., $|\B| \geq q^{NR}rn + 1$, that has a multiplicative sub-group $\E$ of order $r n$, i.e., $|\E| = rn$. Note that the sub-group has $\frac{|\B^{\ast}|}{|\E|} \geq q^{NR}$ cosets, each of size $|\E| = rn$. We associate $q^{NR}$ distinct cosets of the sub-group $\E$ to $q^{NR}$ different codewords of the code $\Cc^{\rm II}$. For a codeword $\cv \in \Cc^{\rm II}$, let $\sigma_{\cv} \in \B^{\ast}$ be such that the coset associated with the codeword $\cv$ is $\sigma_{\cv}\cdot\E$. 

We take $\Cc^{\rm I}$ to be the code obtained from the Ye and Barg construction with $rn$ distinct elements of the multiplicative subgroup $\E$ forming the $rn$ evaluation points $\{\lambda_{i,j}\}_{i \in [n], j \in [0:r-1]}$ (cf.~Appendix~\ref{appen:YeBarg}). Note that the code $\Cc^{\rm I}$ is defined by the parity-check matrix $\Hm$ (cf.~\eqref{eq:Hm}), where {$n$ thick columns of the parity-check matrix $\Hm$} corresponding to $n$ distinct code blocks in a codeword of $\Cc^{\rm I}$ are defined by the $n$ distinct $\ell \times \ell$  matrices $\{A_{1}, A_{2},\ldots, A_{n}\}$. In order to fully specify the code $\Cc$ obtained from Construction~\ref{cons:generic}, we also need to specify the scalar $\{\alpha_{j, \cv}\}_{j \in [r], \cv \in \Cc^{\rm II}}$ (cf.~\eqref{eq:thickCol}). For $j \in [r]$ and $\cv \in \Cc^{\rm II}$, we assign
$$
\alpha_{j, \cv} = \sigma^{j-1}_{\cv},
$$
where, as defined earlier, $\sigma_{\cv}$ specifies the coset of $\E$ which is associated with the codeword $\cv \in \Cc^{\rm II}$.

Let $\Hc$ be the $rN\ell \times MN\ell$ parity-check matrix of the code $\Cc$ obtained from Construction~\ref{cons:generic}. Recall that a codeword of $\Cc$ has $M$ code blocks which are indexed by the codewords of $\Cc^{\rm II}$. Given the aforementioned choice for the short MSR code $\Cc^{\rm I}$, the $N\ell$ columns of $\Hc$ corresponding to the code block indexed by $\cv \in \Cc^{\rm II}$ takes the following form (cf.~\eqref{eq:thickCol}).
\begin{align}
\label{eq:Hm_col}
\Hc_{\cv} = \left(\begin{array}{c}
\Diag(\Id, \Id,\ldots, \Id) \\
\sigma_{\cv}\cdot\Diag(A_{c_1}, A_{c_2},\ldots, A_{c_N}) \\
\sigma^2_{\cv}\cdot\Diag(A^2_{c_1}, A^2_{c_2},\ldots, A^2_{c_N}) \\
\vdots \\
\sigma^{r-1}_{\cv}\cdot\Diag(A^{r-1}_{c_1}, A^{r-1}_{c_2},\ldots, A^{r-1}_{c_N} \\
\end{array} \right) = \left( \begin{array}{c}
\Id  \\
\sigma_{\cv}\cdot\Ac_{\cv} \\
\sigma^2_{\cv}\cdot\Ac^2_{\cv} \\
\vdots\\
\sigma^{r-1}_{\cv}\cdot\Ac^{r -1}_{\cv} \\
\end{array}\right),
\end{align}
where $\Id$ denotes both $\ell \times \ell$  and $N\ell \times N\ell$ identity matrices. Moreover, we use $\Ac_{\cv}$ to denote the following $N \ell \times N \ell$ block diagonal matrix 
\begin{align}
\label{eq:Ac_d}
\Diag(A_{c_1}, A_{c_2},\ldots, A_{c_N}).
\end{align}

\subsubsection{Repair bandwidth for repairing a single code block (node) $\Cc$} As show in the proof of Theorem~\ref{thm:main}, the code block indexed by $\cv \in \Cc^{\rm II}$ in a codeword of $\Cc$ can be repaired using the following $N\ell \times rN\ell$ repair matrix. 
\begin{align}
\label{eq:repairSpecial}
\Sc_{\cv} = \Diag(S_{c_1}, S_{c_2},\ldots, S_{c_N}),
\end{align}
where ${\ell} \times r\ell$ matrices $\{S_{c_i}\}_{i \in [N]}$ are defined in Appendix~\ref{sec:repMat_YeB}. Taking the code $\Cc^{\rm II}$ with large enough distance, the repair bandwidth associated with linear repair scheme defined by these repair matrices can be made at most $(1 + \epsilon)\cdot\frac{N\ell}{r}$ (cf.~Remark~\ref{rem:generic}).

\subsubsection{MDS property of $\Cc$} Next, we argue that the code $\Cc$ obtained in this section is an MDS code. Along with the previous result on its repair bandwidth, the following result establishes that $\Cc$ is an $\epsilon$-MSR code.
\begin{lemma}
Let $\Cc$ be a linear array code defined by the $rN\ell \times q^{NR}N\ell$ parity-check matrix $\Hc$ as described in \eqref{eq:Hm_col}. Then, $\Cc$ is a $[q^{NR}, q^{NR} - r, r+1, N\ell]_{\B}$ MDS code.
\end{lemma}
\begin{proof}
In order to argue that $\Cc$ is an MDS code, we need to show that any $rN\ell \times rN\ell$ sub-matrix of $\Hc$ consisting of $r = n- k$ thick columns of $\Hc$ corresponding to any $r$ distinct code blocks is full rank. Let's consider the $r$ code blocks indexed by the following $r$ codewords of $\Cc^{\rm II}$.
$$
\Pc = \big\{\cv^{1}, \cv^{2},\ldots, \cv^{r}\big\} \subset \Cc^{\rm II}.
$$
The $r N\ell \times r N\ell$ sub-matrix of $\Hc$ that corresponds to the code blocks indexed by these codewords takes the following form.
\begin{align}
\Hc_{\Pc} &=  \left( \begin{array}{cccc}
\Hc_{\cv^{1}} & \Hc_{\cv^{2}} & \cdots & \Hc_{\cv^{r}}
\end{array}\right)  \nonumber  \\
&=\left( \begin{array}{cccc}
\Id & \Id & \cdots & \Id \\
\sigma_{\cv^{1}}\cdot\Ac_{\cv^{1}} & \sigma_{\cv^{2}}\cdot\Ac_{\cv^{2}}  & \cdots & \sigma_{\cv^{r}}\cdot\Ac_{\cv^{r}}  \\
\sigma^2_{\cv^{1}}\cdot\Ac_{\cv^{1}} &  \sigma^2_{\cv^{2}}\cdot\Ac^2_{\cv^{2}}  & \cdots & \sigma^2_{\cv^{r}}\cdot\Ac^2_{\cv^{r}}  \\
\vdots & \vdots & \ddots & \vdots \\
\sigma^{r-1}_{\cv^{1}}\cdot\Ac_{\cv^{1}} & \sigma^{r-1}_{\cv^{2}}\cdot\Ac^{r-1}_{\cv^{2}}  & \cdots & \sigma^{r-1}_{\cv^{r}}\cdot\Ac^{r-1}_{\cv^{r}}  \\
\end{array} \right).
\end{align}
Taking the block diagonal structure of the matrices $\{\Ac_{\cv^{w}}\}_{w \in [r]}$ into account (cf.~\eqref{eq:Ac_d}), it is sufficient to argue that for every $i \in [N]$ the following matrix is full rank.
\begin{align}
\Um_{\Pc, i} = \left( \begin{array}{cccc}
\Id & \Id & \cdots & \Id \\
\sigma_{\cv^{1}}\cdot A_{c^{1}_i} & \sigma_{\cv^{2}}\cdot A_{c^{2}_i}  & \cdots & \sigma_{\cv^{r}}\cdot A_{c^{r}_i}  \\
\sigma^2_{\cv^{1}}\cdot A_{c^{1}_i} &  \sigma^2_{\cv^{2}}\cdot A^2_{c^{2}_i}  & \cdots & \sigma^2_{\cv^{r}}\cdot A^2_{c^{r}_i}  \\
\vdots & \vdots & \ddots & \vdots \\
\sigma^{r-1}_{\cv^{1}}\cdot A_{c^{1}_i} & \sigma^{r-1}_{\cv^{2}}\cdot A^{r-1}_{c^{2}_i}  & \cdots & \sigma^{r-1}_{\cv^{r}}\cdot A^{r-1}_{c^{r}_i}  \\
\end{array} \right),
\end{align}
where $c^{j}_{i}$ denotes the $i$-th code symbol in the codeword $\cv^{j} \in \Pc \subset \Cc^{\rm II}$. For any $i \in [n]$, $\Um_{\Pc, i}$ is a block matrix with diagonal blocks (cf.~\eqref{eq:Am}). Similar to the proof of Theorem III.2  in \cite{YeB16a}, one can rearrange the rows and columns of the matrix $\Um_{\Pc, i}$ to obtain a block {\em diagonal} matrix, where diagonal blocks are Vandermonde matrices. Therefore, the matrix $\Um_{\Pc, i}$ is a full rank matrix. This completes the proof.
\end{proof}

\subsection{Existence of good $\epsilon$-MSR codes using GV bound}
\label{sec:gv}

In Section~\ref{sec:composing}, we present an explicit $\epsilon$-MSR code by using Ye and Barg codes~\cite{YeB16a} as $\Cc^{\rm I}$ and any explicit error correcting code with large enough minimum distance as $\Cc^{\rm II}$ in Construction~\ref{cons:generic}. As highlighted by Theorem~\ref{thm:main}, various parameters of the $\epsilon$-MSR code obtained from Construction~\ref{cons:generic}, including code length and sub-packetization level, depend on the specific $\Cc^{\rm II}$ used in the construction. Next, we utilize the codes operating on the GV bound~\cite{MacSlo} as $\Cc^{\rm II}$ to show that our approach establishes the existence of $\epsilon$-MSR codes that have a good scaling between the code length and the sub-packetization level.

\begin{theorem}
\label{thm:param}
Given an integer $r\geq 1$ and $\epsilon>0$, there exists a constant $s=s(r,\epsilon)>0$ such that for infinite values of $l$ there exists an $(\Nc = \Omega (\exp(sl)), \Kc = \Nc - r, \Tc = \Nc - 1, l)_{\B}$ $\epsilon$-MSR code. Furthermore, the required field size $|\B|$ scales as $O(\Nc)$.
\end{theorem}

\begin{proof}
Recall that it follows from GV bound that for every alphabet of size $q$ and $\delta \in (0, 1/q)$, there exists a code over the alphabet with relative minimum distance at least $\delta$ and  rate 
\begin{align}
\label{eq:GV}
R \geq 1 - h_{q}(\delta) - o(1),
\end{align}
where $h_q(x) = x\log_q(q-1) -x\log_qx - (1-x)\log_q(1-x)$ denotes the $q$-ary entropy function.
For a constant $\epsilon > 0$, we choose $q$ large enough such that $\delta^{\ast} = 1 - \frac{\epsilon}{r-1} < 1 - 1/q$. Now we take $\Cc^{\rm II}$ to be an $N$-length code over an alphabet of size $q$ with code size $q^{N(1 - h_q(\delta^{\ast}) - o(1))}$ and relative minimum distance at least $\delta^{\ast} = 1 - \frac{\epsilon}{r-1}$ (cf.~\eqref{eq:GV}). We combine this with the an $(n = q, k = q-r, t = q-1, \ell = r^q)_{\B}$ MSR code from \cite{YeB16a} as described above. This gives us an $\epsilon$-MSR code with length $\Nc =q^{N(1 - h_q(\delta^{\ast}) - o(1))}$ and sub-packetization level $l = N\ell = Nr^{q}$. Therefore, we have
\begin{align}
\Nc = q^{({(1 - h_q(\delta^{\ast}) - o(1))}/{r^q})l}.
\end{align}
For constant $r$ and $q$, this can be expressed as $\Nc = \Omega(\exp(sl))$ for a suitable constant $s$. Note that for constant $r$ and $q$, the required filed size $q^{N(1 - h_q(\delta^{\ast}) - o(1))}qr + 1$ scales linearly with $\Nc$, the length of the code.
\end{proof}

\begin{remark}
\label{rem:GR17}
Note that the sub-packetization level of the $\epsilon$-MSR codes mentioned in Theorem~\ref{thm:param} satisfy $l = O\big( r^{\left(r/\epsilon\right)} \cdot\log\Nc\big)$. For the identical repair-bandwidth guarantees, the MDS codes corresponding to Theorem~\ref{thm:main_tau} (cf.~Section~\ref{sec:twiceRB} and \ref{sec:RB_t}) require a smaller sub-packetization level of $r^{1/\epsilon}$. However, as compared to Theorem~\ref{thm:main_tau}, the codes mentioned in Theorem~\ref{thm:param} ensure load balancing during the repair process and require field size which is only linear in the code length $\Nc$ when $r = \Theta(1)$. Recall that the codes constructed in Section~\ref{sec:twiceRB} and \ref{sec:RB_t} require field size which is exponential in the code length.
\end{remark}

\section{Necessary Sub-packetization for $\epsilon$-MSR Codes.}
\label{sec:converse}
In Section~\ref{sec:long_mds}, we establish that resorting to $\epsilon$-MSR codes for positive $\epsilon$ allows the number of nodes to scale exponentially with the sub-packetization level $\ell$. This is an encouraging result, since for MSR codes with constant $r = n - k$, it is known that $n$ has to scale polylogarithmically with $\ell$~\cite{TWB14,GTC14}. In this section we derive an upper bound on the number of nodes in a {\em linear} $\epsilon$-MSR code, equivalently a lower bound on $\ell$ as a function of $n$. The bound relies on an approach similar to the one employed in \cite{TWB14,GTC14} to bound the number of nodes in an MSR code. Each node is assigned a vector in some vector space. Then it is shown that the assigned vectors are linearly independent and hence the number of such vectors (and nodes) is at most the dimension of the vector space.   

\begin{theorem}
In an $(n,k, t = n - 1, \ell)_{\B}$ linear $\epsilon-$MSR code, the number of nodes $n$ is upper bounded by
$(r\ell)^{\frac{\ell}{r}(1+\epsilon)+1}$.  
\end{theorem}

\begin{proof}
The proof is similar to the proof of Theorem 4 in \cite{TWB14}. Let the parity-check matrix of the code  be the   $r\ell\times n\ell$  matrix    
\begin{equation}
\label{eq:stam2}
\Hm= \left({\begin{array}{cccc} 
	 \Hm_{1} & \cdots & \Hm_n \\     
\end{array}}\right),
\end{equation}
where each $\Hm_i$ is an $r\ell\times \ell$ matrix. Given that the underlying code is an $\epsilon$-MSR code, there exists an $\ell\times r\ell$ matrix $S_i$ which satisfies  
\begin{equation}
\label{eq:rank1}
\rk (S_i\Hm_i)= \ell, 
\end{equation}
and
\begin{equation}
\rk (S_i\Hm_j) \leq \frac{\ell}{r}(1+\epsilon)~\text{for}~i \neq j.
\label{eq:rank}
\end{equation} 
Since the $\ell\times \ell$ matrix $S_i\Hm_i$ is of full rank there exist a subset of columns $C_i$ of size $u=\frac{\ell}{r}(1+\epsilon)+1$, such that the restriction of $S_i\Hm_i$ to the columns in $C_i$ and the first $u$ rows, is a matrix of  full rank, i.e., $$\rank (S_i\Hm_i)_{[u],C_i}=u.$$
Define for node $i$ the polynomial  $f_i:\mathbb{F}^{\ell\times r\ell} \rightarrow \mathbb{F}$ as 
$$f_i(X)=\det(X\Hm_i)_{[u],C_i},$$
where $X=(x_{i,j})$ is an $\ell\times r\ell$ matrix in the variables $x_{i,j}$. It follows from \eqref{eq:rank1} and \eqref{eq:rank} that
$$f_i(S_j)=
\begin{cases}
\neq 0 & i=j\\
0 &   i \neq j.
\end{cases}
$$ 
This implies that the $n$ polynomials $f_i$ are linearly independent. If we assume the contrary, then there exists scalars $\alpha_i$  not all zeros such that $\sum_{j}\alpha_jf_j=0$. However by plugging $S_i$ on both sides of the equation we get that 
$$0=\sum_{j}\alpha_jf_j(S_j)=\alpha_if_i(S_i).$$ Since $f_i(S_i)\neq 0$ we get that $\alpha_i=0$. By repeating this argument for all $i$'s, we get to a contradiction.

Each polynomial $f_i$ is a homogenous polynomial of degree $u$, spanned by the monomials of the form $\prod_{m=1}^ux_{m,j_m}$.  Clearly, there are $(r\ell)^u$ such monomials, and therefore the polynomials $f_i$ form a set of linearly independent vectors in a vector space of dimension $(r\ell)^u$. The result follows since the size of such set can not exceed the dimension.
\end{proof}

%%%%%%%%%%%%%%%%%%%%%%%%%%%%%%%%%%%%%%%%%%%%%%%%%%%%%%%%%%%%%%%
%%%%%%%%%%%%%%%%%%%%%%%%%%%%%%%%%%%%%%%%%%%%%%%%%%%
%% Conclusion
%%%%%%%%%%%%%%%%%%%%%%%%%%%%%%%%%%%%%%%%%%%%%%%%%%%
%%%%%%%%%%%%%%%%%%%%%%%%%%%%%%%%%%%%%%%%%%%%%%%%%%%%%%%%%%%%%%%
\section{Conclusion and Future Directions.}
\label{sec:conclusion}

We explore the constructions of MDS codes that allow for exact repair of a single code block by downloading near-optimal amount of information from the remaining code blocks. These codes are well suited for distributed storage systems as they have the desirable property of working with small sub-packetization levels. In particular, we present two approaches to construct such codes. First, we present a construction which designs parity-check matrices for exact-repairable MDS codes with small sub-packetization level by utilizing the parity-check matrices of any scalar MDS codes. This approach allow us to barter the repair bandwidth (the cost of repairing a single code block) with the underlying sub-packetization by choosing a design parameter. The codes obtained from this construction also  enable repair-by-transfer mechanisms, where repair of a code block (node) requires  minimal computation at the contacted code blocks (nodes). 

Furthermore, we propose a general approach to transform a short MSR code, i.e.,  an MDS code with optimal repair bandwidth, which has a large sub-packetization level to a long exact-repairable MDS code that has small sub-packetization level (as a function of the code length). The codes obtained in this manner also ensure load balancing among the contacted code blocks during the repair process as all the contacted code blocks contribute similar amount of information. Recognizing that the obtained codes broadly share the load balancing property with the MSR codes, we term these codes as $\epsilon$-MSR codes. In addition, we also note that if the utilized short MSR code has a repair-by-transfer mechanism, then the obtained long $\epsilon$-MSR code would also provide a repair-by-transfer mechanism.

There are multiple interesting directions to extend this work. We conclude by pointing out a few of those directions. In this paper we focus on the setting with $t = n - 1$ which requires contacting all the intact code blocks for the repair of a failed code block. Extending our results for general value of $k < t < n - 1$ is an important direction to explore. For the first construction presented in this paper (cf. Section~\ref{sec:twiceRB} and \ref{sec:RB_t}), one relatively straightforward approach is to employ the ideas used in \cite{RKV16a} to extend the construction from \cite{SAK15} to general values of $t$. For an integer $\tau \geq 1$, this would give exact-repairable codes with sub-packetization level $(t - k + 1)^{\tau}$ and small repair bandwidth. Moreover, the obtained codes would also have repair-by-transfer schemes. We believe that our second construction (cf.~Section~\ref{sec:gen_approach}) can also be modified for the setting with $t < n - 1$. This would require utilizing a short MSR code that can support repair of each of its code blocks without contacting all the remaining code blocks.

The problem of designing MDS codes that allow for simultaneous repair of multiple code blocks has been addressed in several works, including \cite{ShumHu, Kermarrec:Repairing11, RKV16b, YeB16a, DDKM16, BW17}. Designing codes that provide mechanisms to perform simultaneous repair of multiple code blocks, and as well as a good trade-off between the sub-packetization level and repair bandwidth is another interesting direction to pursue. 

We also present a lower bound on the sub-packetization level that is necessary for an $\epsilon$-MSR code. However, there is a gap between this bound and the sub-packetization levels achieved by our constructions. Finally, we note that the exact-repairability and the corresponding repair bandwidth of the codes obtained by our first construction (cf.~Section~\ref{sec:twiceRB} and \ref{sec:RB_t}) only depend on the combinatorial structure, i.e., the locations of non-zero entries, of the designed parity-check matrix. However, the argument which establishes the MDS property for these codes requires the size of the base field $\B$ to be quite large. The similar issue also arises in many previous works, e.g.,~\cite{zigzag13, SAK15}. Recently, Ye and Barg have addressed this issue for the codes that operated exactly at the cut-set bound in \cite{YeB16a, YeB16b}. However, they again work with large sub-packetization level $n^{\ceilbb{\frac{n}{n-k}}}$. The reduction of the base field size for our first construction is an interesting question, which has both theoretical and practical significance.

%%%%%%%%%%%%%%%%%%%%%%%%%%%%%%%%%%%%%%%%%%%%%%%%%%%%%%%%%%%%%%%
%%%%%%%%%%%%%%%%%%%%%%%%%%%%%%%%%%%%%%%%%%%%%%%%%%%
%% References
%%%%%%%%%%%%%%%%%%%%%%%%%%%%%%%%%%%%%%%%%%%%%%%%%%%
%%%%%%%%%%%%%%%%%%%%%%%%%%%%%%%%%%%%%%%%%%%%%%%%%%%%%%%%%%%%%%% 

\bibliographystyle{plain}
\bibliography{RepairBW_MDS}

\begin{thebibliography}{10}

\bibitem{BW17}
B.~Bartan and M.~Wootters.
\newblock Repairing multiple failures for scalar {MDS} codes.
\newblock {\em CoRR}, abs/1707.02241, 2017.

\bibitem{BlaumHH13}
M.~Blaum, J.~L. Hafner, and S.~Hetzler.
\newblock Partial-{MDS} codes and their application to {RAID} type of
  architectures.
\newblock {\em IEEE Transactions on Information Theory}, 59(7):4510--4519, July
  2013.

\bibitem{Cadambe_poly}
V.~R. Cadambe, C.~Huang, J.~Li, and S.~Mehrotra.
\newblock Polynomial length {MDS} codes with optimal repair in distributed
  storage.
\newblock In {\em Proc. of Forty Fifth Asilomar Conference on Signals, Systems
  and Computers (ASILOMAR)}, pages 1850--1854, Nov 2011.

\bibitem{CJMRS13}
V.~R. Cadambe, S.~A. Jafar, H.~Maleki, K.~Ramchandran, and C.~Suh.
\newblock Asymptotic interference alignment for optimal repair of {MDS} codes
  in distributed storage.
\newblock {\em IEEE Transactions on Information Theory}, 59(5):2974--2987, May
  2013.

\bibitem{DDKM16}
S.~H. Dau, I.~M. Duursma, H.~M. Kiah, and O.~Milenkovic.
\newblock Repairing reed-solomon codes with multiple erasures.
\newblock {\em CoRR}, abs/1612.01361, 2016.

\bibitem{DM17}
S.~H. Dau and O.~Milenkovic.
\newblock Optimal repair schemes for some families of full-length reed-solomon
  codes.
\newblock {\em CoRR}, abs/1701.04120, 2017.

\bibitem{dimakis}
A.~G. Dimakis, P.~Godfrey, Y.~Wu, M.~Wainwright, and K.~Ramchandran.
\newblock Network coding for distributed storage systems.
\newblock {\em IEEE Transactions on Information Theory}, 56(9):4539--4551, Sept
  2010.

\bibitem{DRWS2011}
A.~G. Dimakis, K.~Ramchandran, Y.~Wu, and C.~Suh.
\newblock A survey on network codes for distributed storage.
\newblock {\em Proc. of the IEEE}, 99(3):476--489, March 2011.

\bibitem{GoparajuFV16}
A.~Fazeli, S.~Goparaju, and A.~Vardy.
\newblock Minimum storage regenerating codes for all parameters.
\newblock In {\em Proc. of 2016 IEEE International Symposium on Information
  Theory (ISIT)}, pages 76--80, July 2016.

\bibitem{Gopalan14}
P.~Gopalan, C.~Huang, B.~Jenkins, and S.~Yekhanin.
\newblock Explicit maximally recoverable codes with locality.
\newblock {\em IEEE Transactions on Information Theory}, 60(9):5245--5256, Sept
  2014.

\bibitem{Gopalan12}
P.~Gopalan, C.~Huang, H.~Simitci, and S.~Yekhanin.
\newblock On the locality of codeword symbols.
\newblock {\em IEEE Transactions on Information Theory}, 58(11):6925--6934, Nov
  2012.

\bibitem{GTC14}
S.~Goparaju, I.~Tamo, and R.~Calderbank.
\newblock An improved sub-packetization bound for minimum storage regenerating
  codes.
\newblock {\em IEEE Transactions on Information Theory}, 60(5):2770--2779, May
  2014.

\bibitem{GR17}
V.~Guruswami and A.~S. Rawat.
\newblock {MDS} code constructions with small sub-packetization and
  near-optimal repair bandwidth.
\newblock In {\em Proc. of the Twenty-Eighth Annual ACM-SIAM Symposium on
  Discrete Algorithms (SODA)}, pages 2109--2122, Philadelphia, PA, USA, 2017.
  Society for Industrial and Applied Mathematics.

\bibitem{GW15}
V.~Guruswami and M.~Wootters.
\newblock Repairing {Reed-solomon} codes.
\newblock In {\em Proc. of the Forty-eighth Annual ACM Symposium on Theory of
  Computing (STOC)}, pages 216--226, New York, NY, USA, 2016. ACM.

\bibitem{HLKB15}
W.~Huang, M.~Langberg, J.~Kliewer, and J.~Bruck.
\newblock Communication efficient secret sharing.
\newblock {\em IEEE Transactions on Information Theory}, 62(12):7195--7206, Dec
  2016.

\bibitem{KPLK12}
G.~M. Kamath, N.~Prakash, V.~Lalitha, and P.~V. Kumar.
\newblock Codes with local regeneration and erasure correction.
\newblock {\em IEEE Transactions on Information Theory}, 60(8):4637--4660, Aug
  2014.

\bibitem{Kermarrec:Repairing11}
A.-M. Kermarrec, N.~Le~Scouarnec, and G.~Straub.
\newblock Repairing multiple failures with coordinated and adaptive
  regenerating codes.
\newblock In {\em Proc. of 2011 International Symposium on Network Coding
  (NetCod)}, pages 1--6, July 2011.

\bibitem{Khan12}
O.~Khan, R.~Burns, J.~Plank, W.~Pierce, and C.~Huang.
\newblock Rethinking erasure codes for cloud file systems: Minimizing {I/O} for
  recovery and degraded reads.
\newblock In {\em Proc. of 10th USENIX Conference on File and Storage
  Technologies (FAST)}, Berkeley, CA, USA, Feb 2012. USENIX Association.

\bibitem{KGJO16}
K.~Kralevska, D.~Gligoroski, R.~E. Jensen, and H.~{\O}verby.
\newblock Hashtag erasure codes: {F}rom theory to practice.
\newblock {\em CoRR}, abs/1609.02450, 2016.

\bibitem{KGO16}
K.~Kralevska, D.~Gligoroski, and H.~{\O}verby.
\newblock General sub-packetized access-optimal regenerating codes.
\newblock {\em IEEE Communications Letters}, 20(7):1281--1284, July 2016.

\bibitem{MacSlo}
F.~J. MacWilliams and N.~J.~A. Sloane.
\newblock {\em {T}he Theory of Error-Correcting Codes}.
\newblock Amsterdam: North-Holland, 1983.

\bibitem{PapDim12}
D.~S. Papailiopoulos and A.~G. Dimakis.
\newblock Locally repairable codes.
\newblock {\em IEEE Transactions on Information Theory}, 60(10):5843--5855, Oct
  2014.

\bibitem{PapDimCad13}
D.~S. Papailiopoulos, A.~G. Dimakis, and V.~Cadambe.
\newblock Repair optimal erasure codes through hadamard designs.
\newblock {\em IEEE Transactions on Information Theory}, 59(5):3021--3037, May
  2013.

\bibitem{RSK11}
K.~Rashmi, N.~Shah, and P.~Kumar.
\newblock Optimal exact-regenerating codes for distributed storage at the {MSR}
  and {MBR} points via a product-matrix construction.
\newblock {\em IEEE Transactions on Information Theory}, 57:5227--5239, Aug
  2011.

\bibitem{RSR13}
K.~V. Rashmi, N.~B. Shah, and K.~Ramchandran.
\newblock A piggybacking design framework for read-and download-efficient
  distributed storage codes.
\newblock In {\em Proc. of 2013 IEEE International Symposium on Information
  Theory (ISIT)}, pages 331--335, July 2013.

\bibitem{RKSV12}
A.~S. Rawat, O.~O. Koyluoglu, N.~Silberstein, and S.~Vishwanath.
\newblock Optimal locally repairable and secure codes for distributed storage
  systems.
\newblock {\em IEEE Transactions on Information Theory}, 60(1):212--236, Jan
  2014.

\bibitem{RKV16b}
A.~S. Rawat, O.~O. Koyluoglu, and S.~Vishwanath.
\newblock Centralized repair of multiple node failures with applications to
  communication efficient secret sharing.
\newblock {\em CoRR}, abs/1603.04822, 2016.

\bibitem{RKV16a}
A.~S. Rawat, O.~O. Koyluoglu, and S.~Vishwanath.
\newblock Progress on high-rate {MSR} codes: Enabling arbitrary number of
  helper nodes.
\newblock {\em CoRR}, abs/1601.06362, 2016.

\bibitem{RTGE17}
A.~S. Rawat, I.~Tamo, V.~Guruswami, and K.~Efremenko.
\newblock $\epsilon$-{MSR} codes with small sub-packetization.
\newblock {\em {\em to appear in} 2017 IEEE International Symposium on
  Information Theory (ISIT)}, 2017.

\bibitem{SAK15}
B.~Sasidharan, G.~K. Agarwal, and P.~V. Kumar.
\newblock A high-rate {MSR} code with polynomial sub-packetization level.
\newblock In {\em Proc. of 2015 IEEE International Symposium on Information
  Theory (ISIT)}, pages 2051--2055, June 2015.

\bibitem{SVK16}
B.~Sasidharan, M.~Vajha, and P.~V. Kumar.
\newblock An explicit, coupled-layer construction of a high-rate {MSR} code
  with low sub-packetization level, small field size and all-node repair.
\newblock {\em CoRR}, abs/1607.07335, 2016.

\bibitem{SPDG14}
K.~Shanmugam, D.~S. Papailiopoulos, A.~G. Dimakis, and G.~Caire.
\newblock A repair framework for scalar {MDS} codes.
\newblock {\em IEEE Journal on Selected Areas in Communications},
  32(5):998--1007, May 2014.

\bibitem{ShumHu}
K.~W. Shum and Y.~Hu.
\newblock Cooperative regenerating codes.
\newblock {\em IEEE Transactions on Information Theory}, 59(11):7229--7258, Nov
  2013.

\bibitem{TamoBarg14}
I.~Tamo and A.~Barg.
\newblock A family of optimal locally recoverable codes.
\newblock {\em IEEE Transactions on Information Theory}, 60(8):4661--4676, Aug
  2014.

\bibitem{TK16}
I.~Tamo and K.~Efremenko.
\newblock New results on {MSR} codes.
\newblock In {\em Information Theory and Applications Workshop (ITA), 2016},
  Feb 2016.

\bibitem{zigzag13}
I.~Tamo, Z.~Wang, and J.~Bruck.
\newblock Zigzag codes: {MDS} array codes with optimal rebuilding.
\newblock {\em IEEE Transactions on Information Theory}, 59(3):1597--1616,
  March 2013.

\bibitem{TWB14}
I.~Tamo, Z.~Wang, and J.~Bruck.
\newblock Access versus bandwidth in codes for storage.
\newblock {\em IEEE Transactions on Information Theory}, 60(4):2028--2037,
  April 2014.

\bibitem{TYB17}
I.~Tamo, M.~Ye, and A.~Barg.
\newblock Optimal repair of {R}eed-{S}olomon codes: {A}chieving the cut-set
  bound.
\newblock {\em CoRR}, abs/1706.00112, 2017.

\bibitem{WDB10}
Z.~Wang, A.~G. Dimakis, and J.~Bruck.
\newblock Rebuilding for array codes in distributed storage systems.
\newblock In {\em Proc. of 2010 IEEE Globecom Workshops}, pages 1905--1909, Dec
  2010.

\bibitem{zigzag_allerton11}
Z.~Wang, I.~Tamo, and J.~Bruck.
\newblock On codes for optimal rebuilding access.
\newblock In {\em Proc. of the 49th Annual Allerton Conference on
  Communication, Control, and Computing (Allerton)}, pages 1374--1381, Sept
  2011.

\bibitem{WTB12}
Z.~Wang, I.~Tamo, and J.~Bruck.
\newblock Long {MDS} codes for optimal repair bandwidth.
\newblock In {\em Proc. of 2012 IEEE International Symposium on Information
  Theory (ISIT)}, pages 1182--1186, July 2012.

\bibitem{WTB16}
Z.~Wang, I.~Tamo, and J.~Bruck.
\newblock Explicit minimum storage regenerating codes.
\newblock {\em IEEE Transactions on Information Theory}, 62(8):4466--4480, Aug
  2016.

\bibitem{XCode14}
S.~Xu, R.~Li, P.~P.~C. Lee, Y.~Zhu, L.~Xiang, Y.~Xu, and J.~C.~S. Lui.
\newblock Single disk failure recovery for {X}-code-based parallel storage
  systems.
\newblock {\em IEEE Transactions on Computers}, 63(4):995--1007, April 2014.

\bibitem{YeBargRS}
M.~Ye and A.~Barg.
\newblock Explicit constructions of {MDS} array codes and {RS} codes with
  optimal repair bandwidth.
\newblock In {\em Proc. of 2016 IEEE International Symposium on Information
  Theory (ISIT)}, pages 1202--1206, July 2016.

\bibitem{YeB16b}
M.~Ye and A.~Barg.
\newblock Explicit constructions of optimal-access {MDS} codes with nearly
  optimal sub-packetization.
\newblock {\em CoRR}, abs/1605.08630, 2016.

\bibitem{YeB16a}
M.~Ye and A.~Barg.
\newblock Explicit constructions of high-rate {MDS} array codes with optimal
  repair bandwidth.
\newblock {\em IEEE Transactions on Information Theory}, 63(4):2001--2014,
  April 2017.

\end{thebibliography}

%%%%%%%%%%%%%%%%%%%%%%%%%%%%%%%%%%%%%%%%%%%%%%%%%%%%%%%%%%%%%%%
%%%%%%%%%%%%%%%%%%%%%%%%%%%%%%%%%%%%%%%%%%%%%%%%%%%

\appendices

\section{Necessary sub-packetization level for MDS codes}
\label{appen:appen_sub}

Let $\Cc \subseteq \F^{n}$ be an MDS code with the sub-packetization level $\ell$ where all contacted nodes contribute to the repair process. For a constant $b \geq 1$, let the repair bandwidth of $\Cc$ for exact repair of a single code block is less than $b$ times the cut-set bound, i.e., number of symbols (over $\B$) downloaded from the contacted $t = n-1$ nodes is at most
\begin{align}
 b\left(\frac{n - 1}{n - k}\right)\cdot \ell. \nonumber
\end{align}
This implies that there exists at least one contacted node which contributes at most $\floorbb{\frac{b\ell}{n - k} }$ symbols (over $\B$) during the repair process. Moreover, each of the contacted $t = n - 1$ nodes sends at least $1$ symbol of $\B$ during the repair process. Hence, we have that
\begin{align}
\floorb{\frac{b\ell}{n - k}}\geq 1. \nonumber
\end{align}
This gives us that 
\begin{align}
\ell \geq \frac{n - k}{b}~~\text{or}~~\ell = \Omega(n-k). \nonumber
\end{align}

\section{Proof of Proposition~\ref{prop:perturbedMDS}.}
\label{appen:perturbedMDS}

\begin{proof}
As illustrated in Example~\ref{ex:ex1}, Construction~\ref{subsec:construction2} obtains the parity-check matrix $\Pm$ by perturbing a parity-check matrix of an MDS code. In particular, we have 
\begin{align}
\label{eq:Pm_mds}
\Pm = \Hm + \Em^{\psi},
\end{align}
where $\Hm$ is a parity-check matrix of an $[n, k\ell = k(n-k), d_{\min} = n - k + 1, \ell = n-k]_{\B}$ MDS code. Furthermore, we select the element $\psi$ such that it generates the field $\B = \Lf(\psi)$ with the extension degree $[\B:\Lf] \geq (r-1)\ell + 1$. Let $m_{\psi}(X) \in \Lf[X]$ be the minimal polynomial $\psi$. It follows from the particular choice of $\psi$ that the degree of $m_{\psi}(X)$ is at least $(r-1)\ell + 1$.

Now, we argue that for such a choice of $\psi$ and the associated field $\B$, the parity-check matrix $\Pm$ (cf.~\eqref{eq:Pm_mds}) defines an $[n, k\ell, d_{\min} = n - k + 1, \ell]_{\B}$ MDS code, i.e., for each $\Sc \subseteq [n]~\text{such that}~|\Sc| = r = n-k$ we have 
\begin{align}
\label{eq:Pm_Sn}
\det(\Pm_{\Sc}) = \det(\Hm_{\Sc} + \Em^{\psi}_{\Sc}) \neq 0.
\end{align}

Recall that $\Hm$ is an $r\ell \times n\ell$ parity-check matrix of an $[n, k\ell, d_{\min} = n - k + 1, \ell]_{\B}$ MDS code. Therefore, we have 
\begin{align}
\label{eq:HmSn}
\det\big(\Hm_{\Sc}\big) \neq 0~~~\forall~\Sc \subseteq [n]~\text{such that}~|\Sc| = r.
\end{align}

Consider a set $\Sc \subseteq [n]$ such that $|\Sc| = r = n - k$ and the associated sub-matrix $\Hm_{\Sc} + \Em^{X}_{\Sc}$ (cf.~\eqref{eq:Pm_mds}), where $X$ is an indeterminate. Note that the determinant of this $r\ell \times r\ell$ matrix can be expressed as
\begin{align}
\label{eq:f_defn}
f_{\Sc}(X) = \det\big(\Hm_{\Sc} + \Em^{X}_{\Sc} \big),
\end{align}
where $f_{\Sc}(X) \in \Lf[X]$ is a polynomial of degree at most $(r-1)\ell$ and its coefficients are defined by the elements $\{\lambda_{i}\}_{i \in \Sc} \subseteq \Lf$. We now argue that the polynomial $f_{\Sc}(X)$ is a non-trivial (not an identically zero) polynomial. Towards this, we consider the value of the polynomial $f_{\Sc}(X)$ at $X = 0$.
\begin{align}
f_{\Sc}(0) &= \det\big(\Hm_{\Sc} + \Em^{0}_{\Sc}\big) \overset{(i)}{=} \det\big(\Hm_{\Sc}\big) \overset{(ii)}{\neq} 0. \nonumber
\end{align}
Here, $(i)$ holds as $\Em^{0}$ reduces to a zero matrix, and $(ii)$ follows from \eqref{eq:HmSn}. Since $f_{\Sc}(X)$ evaluates to a non-zero value at $X = 0$, it's a non-trivial polynomial. We now substitute $X = \psi$, which gives us the following (cf.~\ref{eq:f_defn}).
\begin{align}
f_{\Sc}(\psi) &= \det\big(\Hm_{\Sc} + \Em^{\psi}_{\Sc} \big) \overset{(i)}{=} \det\big(\Pm_{\Sc}\big) \overset{(ii)}{\neq} 0. \nonumber
\end{align}
Here, $(i)$ follows from the definition of $\Pm$ (cf.~\eqref{eq:Pm_mds} and \eqref{eq:Pm_Sn}). The step $(ii)$ follows as we have that the degree of $f_{\Sc}(X) \in \Lf[X]$ is strictly less than the degree of $m_{\psi}(X) \in \Lf[X]$, the minimal polynomial of $\psi$. Since the choice of $\Sc$ is arbitrary over all the subsets of $[n]$ of size $r = n-k$. We have that 
\begin{align}
\nonumber
f_{\Sc}(\psi) = \det\big(\Pm_{\Sc}\big) \neq 0~\forall\Sc \subseteq [n]~\text{such that}~|\Sc| = r.
\end{align}
This completes the proof. 
\end{proof}

\section{Ye and Barg Construction~\cite{YeB16a}}
\label{appen:YeBarg}

Let $\B$ be a field with $|\B| \geq rn$ and $\E = \big\{\lambda_{i, j}\big\}_{i \in [n], j \in [0:r - 1]}$ be a set of $rn$ distinct elements of the field $\B$. For $i \in [n]$, consider a matrix $A_i$ which is defined as follows. 
\begin{align}
\label{eq:Am}
A_{i}  = \sum_{a = 0}^{\ell - 1}\lambda_{i, a_i}\ev_{a}\ev^T_{a} =  \sum_{a = 0}^{r^{n} - 1}\lambda_{i, a_i}\ev_{a}\ev^T_{a} \in \B^{\ell \times \ell},
\end{align}
where $\{\ev_{a}\}_{a \in [0:r^{n}-1]}$ denotes the collection of $\ell = (n-k)^n =  r^{n}$ standard basis vectors of $\B^{\ell}$, i.e., all but the $a$-th coordinate of the vector $\ev_{a}$ are equal to $0$ and the $a$-th coordinate has its entry equal to $1$.  Given the $n$ matrices $\{A_1, A_2,\ldots, A_n\}$, let $\Cc(\E) \subseteq \B^{n\ell}$ denote the MSR code defined by the following $r\ell \times n\ell$ parity-check matrix with a Vandermonde type structure.
\begin{align}
\label{eq:Hm}
\Hm = \left( \begin{array}{cccc}
\Id & \Id & \cdots & \Id\\
A_{1} & A_{2} & \cdots & A_{n} \\
A^2_{2} & A^2_{2} & \cdots & A^2_{n} \\
\vdots & \vdots & \ddots & \vdots \\
A^{r-1}_{1} & A^{r-1}_{2} & \cdots & A^{r-1}_{n} \\
\end{array} \right) \in \B^{r\ell \times n\ell},
\end{align}
where $\Id$ denotes the $\ell \times \ell$ identity matrix. 

\subsection{Single node repair in Ye and Barg construction}
\label{appen:repairYeB}

Let the $i$-th code block $\cv_i = \big(c_{i, 0}, c_{i, 1},\ldots, c_{i, \ell-1}\big)$ be the code block being repaired. Note that all the block in the parity-check matrix $\Hm$ (cf.~\eqref{eq:Hm}) are diagonal matrices (cf.~\eqref{eq:Am}). Therefore, the parity-check constraints defining the code $\Cc(\E)$ can be rewritten as follows. 
\begin{align}
\label{eq:code_def1}
\sum_{i = 1}^n\lambda^w_{i, a_i}c_{i, a} = 0~~~~~\forall~~w = [0 : n -k - 1]~\text{and}~a = [0: \ell - 1].
\end{align} 
The repair mechanism recovers the $r = n-k$ symbols $$\big\{c_{i, a(i, 0)}, c_{i, a(i, 1)},\ldots, c_{i, a(i, r -1)}\big\}$$ with the help of the following set of symbols downloaded from the remaining $t = n - 1$ code blocks. 
\begin{align}
\mu^{(a)}_{j, i} := \sum_{u = 0}^{r-1}c_{j, a(i, u)},~~~~j \in [n]\backslash \{i\}.
\end{align}
In particular, for $a \in \{0, 1,\ldots, \ell - 1\}$ and $u \in \{0, 1,\ldots, r - 1\}$, it follows from \eqref{eq:code_def1} that
\begin{align}
\label{eq:code_def2}
\lambda^w_{i, u}c_{i, a(i, u)} + \sum_{j \neq i}\lambda^w_{j, a_j}c_{j, a(i, u)} = 0.
\end{align}
Summing \eqref{eq:code_def2} over $u = 0, 1,\ldots, r-1$, we obtain the following system of equations.
\begin{align}
\left( \begin{array}{cccc}
1 & 1 & \cdots & 1 \\
\lambda_{i, 0} & \lambda_{i,1} & \cdots & \lambda_{i, r-1}\\
\lambda^2_{i, 0} & \lambda^2_{i, 1} & \cdots & \lambda^2_{i, r - 1} \\
\vdots & \vdots & \ddots & \vdots \\
\lambda^{r-1}_{i, 0} & \lambda^{r-1}_{i, 1} & \cdots & \lambda^{r-1}_{i, r - 1} \\
\end{array} \right) 
\left( \begin{array}{c} 
c_{i, a(i, 0)} \\
c_{i, a(i, 1)} \\
c_{i, a(i, 2)} \\
\vdots\\
c_{i, a(i, r-1)}\\
\end{array} \right) = - \left( \begin{array}{c}
\sum_{u = 0}^{r-1}\sum_{j \neq i} c_{j, a(i,u)} \\
\sum_{u = 0}^{r-1}\sum_{j \neq i} \lambda_{j, a_j} c_{j, a(i,u)}\\
\sum_{u = 0}^{r-1}\sum_{j \neq i} \lambda^2_{j, a_j} c_{j, a(i,u)} \\
\vdots \\
\sum_{u = 0}^{r-1}\sum_{j \neq i} \lambda^{r-1}_{j, a_j} c_{j, a(i,u)} \\
\end{array} \right), \nonumber 
\end{align}
or
\begin{align}
\left( \begin{array}{cccc}
1 & 1 & \cdots & 1 \\
\lambda_{i, 0} & \lambda_{i,1} & \cdots & \lambda_{i, r-1}\\
\lambda^2_{i, 0} & \lambda^2_{i, 1} & \cdots & \lambda^2_{i, r - 1} \\
\vdots & \vdots & \ddots & \vdots \\
\lambda^{r-1}_{i, 0} & \lambda^{r-1}_{i, 1} & \cdots & \lambda^{r-1}_{i, r - 1} \\
\end{array} \right) 
\left( \begin{array}{c} 
c_{i, a(i, 0)} \\
c_{i, a(i, 1)} \\
c_{i, a(i, 2)} \\
\vdots\\
c_{i, a(i, r-1)}\\
\end{array} \right) = - \left( \begin{array}{c}
\sum_{j \neq i} \mu^{(a)}_{j,i} \\
\sum_{j \neq i} \lambda_{j, a_j} \mu^{(a)}_{j,i}\\
\sum_{j \neq i} \lambda^2_{j, a_j}\mu^{(a)}_{j,i} \\
\vdots \\
\sum_{j \neq i} \lambda^{r-1}_{j, a_j}\mu^{(a)}_{j,i} \\
\end{array} \right).
\end{align}
Since $\{\lambda_{i, 0}, \lambda_{i,1},\ldots, \lambda_{i, r-1} \}$ are all distinct elements, this system of equations can be solved for the desired code symbols
$
\big\{c_{i, a(i, 0)}, c_{i, a(i, 1)},\ldots, c_{i, a(i, r -1)}\big\}.
$

\subsection{Repair matrices corresponding to the repair process}
\label{sec:repMat_YeB}
This linear repair scheme for $\Cc(\E)$ as described in Appendix~\ref{appen:YeBarg} can be expressed in the form repair matrices (cf.~Section~\ref{sec:linearRepair}). For $i \in [n]$, the $\ell \times r\ell$ repair matrix enabling repair of the $i$-th code block takes the following special block diagonal form with identical $\frac{\ell}{r}\times \ell$ sized diagonal blocks.
\begin{align}
\label{eq:SpecialS}
S_i = \Id \otimes D_i, 
\end{align}
where $\Id$ denotes the $r \times r$ identity matrix. The rows and columns of the $\frac{\ell}{r} \times \ell$ matrix $D_i$ are indexed by the sets $[0:r^{n - 1} - 1]$ and $[0:r^n - 1]$, respectively. 
For $b \in [0:r^{n - 1} - 1]$ and $a \in [0:r^n - 1]$, let $(b_{n-1}, b_{n-2},\ldots, b_1)\in [0:r-1]^{n-1}$ and $(a_{n}, a_{n - 1},\ldots,a_1) \in [0:r - 1]^{n}$ denote their $r$-ary vector representations, respectively. With this notation in place, the $(b,a)$-th entry of $D_i$ is
\begin{align}
D_i(b, a) = \begin{cases} 1 & \text{if}~~~~a_{\backslash\{i\}} = b, \\
0 & \text{otherwise}.
\end{cases}
\end{align}
Here, $a_{\backslash\{i\}} = b$, if we have
$$
(b_{n-1}, b_{n-2},\ldots, b_1) = (a_{n}, a_{n-1},\ldots, a_{i+1}, a_{i-1},\ldots, a_1).
$$
Note that each of the $\frac{\ell}{r} = r^{n-1}$ rows of the matrix $D_i$ has exactly $r  = n-k$ non-zero entries. For $b \in [0:r^{n - 1} - 1]$, $a \in [0:r^n - 1]$ and $w \in [r-1]$, we have that
\begin{align}
D_iA^w_{i}(b, a) = \begin{cases}
\lambda^w_{i, a_i} & \text{if}~~~~a_{\backslash\{i\}} = b, \\
0 & \text{otherwise}.
\end{cases}
\end{align}
and
\begin{align}
D_iA^w_{j}(b, a) = \begin{cases}
\lambda^w_{i, a_j} & \text{if}~~~~a_{\backslash\{i\}} = b, \\
0 & \text{otherwise}.
\end{cases}
\end{align}
Now, for any distinct $w_1, w_2 \in [0:r-1]$, it is straightforward to verify the following.
\begin{align}
\label{eq:IA_r}
D_iA^{w_1}_{j} \cap D_iA^{w_2}_{j} = \begin{cases}
\{0\} & \text{if}~~i = j,\\
D_i & \text{otherwise},
\end{cases}
\end{align} 
where by abuse of notation we use the matrices to denote the subspaces spanned by their rows. For the underlying parity-check matrix $\Hm$ (cf.~\eqref{eq:Hm}) and repair matrix $S_i$ (cf.~\eqref{eq:SpecialS}), two kind of matrices involved in the linear repair scheme takes the following form (cf.~\eqref{eq:Repair_cond1} \& \eqref{eq:Repair_cond2}).
\begin{align}
\label{eq:Repair_cond1S}
\rank\left(S_i \left[\begin{array}{c}
H_{1, i} \\
\vdots \\
H_{r, i}
\end{array}
\right] \right) =  \rank\left(\left[\begin{array}{c}
D_i \\
D_iA_{i} \\
\vdots \\
D_iA^{r-1}_{i}
\end{array}
\right] \right) \overset{(i)}{=} \ell,
\end{align}
and
\begin{align}
\label{eq:Repair_cond2S}
\sum_{j \in [n]\backslash\{i\}}\rank\left(S_i\left[\begin{array}{c}
H_{1, j} \\
\vdots \\
H_{r, j}
\end{array}
\right] \right)  = \sum_{j \in [n]\backslash\{i\}}\rank\left(\left[\begin{array}{c}
D_i \\
D_iA_{j} \\
\vdots \\
D_iA^{r-1}_{j}
\end{array}
\right] \right) \overset{(ii)}{=} (n-1)\rank(D_i) \leq (n - 1)\frac{\ell}{r},
\end{align}
where $(i)$ and $(ii)$ follow from \eqref{eq:IA_r}.

%%%%%%%%%%%%%%%%%%%%%%%%%%%%%%%%%%%%%%%%%%%%%%%%%%%%%%%%%%%%%%%
%%%%%%%%%%%%%%%%%%%%%%%%%%%%%%%%%%%%%%%%%%%%%%%%%%%

\end{document}